\documentclass[runningheads]{llncs}

\usepackage{amssymb, amsmath, hyperref}
\usepackage{bussproofs, stmaryrd, synttree, scrextend, rotating, xcolor}
\usepackage{esvect}
\usepackage{pgf, tikz, color, comment}
\usepackage[all]{xy}
\usepackage{changepage}
\usepackage{enumerate}
\usepackage{proof,upgreek}

\newcommand\seq[4]{{#1} \sslash {#2} \vdash {#3} \sslash {#4}}

\def\imp{\supset}

\def\uc{\overline{\forall}}

\def\cut{(cut)}
\def\iw{(iw)}
\def\ew{(ew)}
\def\icl{(ic_{l})}
\def\icr{(ic_{r})}

\def\mrg{(mrg)}
\def\lwr{(lwr)}
\def\idone{(id_{1})}
\def\idtwo{(id_{2})}
\def\botl{(\bot_{l})}
\def\landl{(\land_{l})}
\def\landr{(\land_{r})}
\def\lorl{(\lor_{l})}
\def\lorr{(\lor_{r})}
\def\impl{(\supset_{l})}
\def\improne{(\supset_{r1})}
\def\imprtwo{(\supset_{r2})}
\def\allrone{(\forall_{r1})}
\def\allrtwo{(\forall_{r2})}
\def\alll{(\forall_{l})}

\def\existsl{(\exists_{l})}
\def\existsr{(\exists_{r})}

\def\lift{(lift)}
\def\vsub{(sub)}

\def\gl{\mathsf{IF}}

\def\ifl{\mathsf{IF}}

\def\calc{\mathsf{LNIF}}

\def\lng{\mathsf{LNG}}

\def\g{\mathcal{G}}
\def\h{\mathcal{H}}

\def\f{\mathcal{F}}

\def\k{\mathcal{K}}

\newcommand{\lang}{\mathcal{L}}

\def\slash{\sslash}
\def\lng{\mathsf{LNG}}

\def\fvx{a} 
\def\fvy{b} 
\def\fvz{c} 

\providecommand{\acknowledgments}[1]{\textbf{Acknowledgments. } #1}

\newtheorem{innercustomlem}{Lemma}
\newenvironment{customlem}[1]
  {\renewcommand\theinnercustomlem{#1.}\innercustomlem}
  {\endinnercustomlem}
  
  \newtheorem{innercustomthm}{Theorem}
\newenvironment{customthm}[1]
  {\renewcommand\theinnercustomthm{#1.}\innercustomthm}
  {\endinnercustomthm}

\begin{document}

\title{Syntactic Cut-Elimination for Intuitionistic Fuzzy Logic via Linear Nested Sequents\thanks{Work funded by FWF projects I2982, Y544-N2, and W1255-N23.}}

\author{Tim Lyon}

\institute{Institut f\"ur Logic and Computation, Technische Universit\"at Wien, 1040 Wien, Austria  \\ \email{lyon@logic.at}}

\titlerunning{Syntactic Cut-Elimination for Intuitionistic Fuzzy Logic}
\authorrunning{Tim Lyon}

\maketitle

\begin{abstract}
This paper employs the \emph{linear nested sequent} framework to design a new cut-free calculus ($\calc$) for intuitionistic fuzzy logic---the first-order G\"odel logic characterized by linear relational frames with constant domains. Linear nested sequents---which are nested sequents restricted to linear structures---prove to be a well-suited proof-theoretic formalism for intuitionistic fuzzy logic. We show that the calculus $\calc$ possesses highly desirable proof-theoretic properties such as invertibility of all rules, admissibility of structural rules, and syntactic cut-elimination. 


\keywords{Cut-elimination · Fuzzy logic · G\"odel logic · Intermediate logic · Intuitionistic logic · Linear nested sequents · Proof theory} 
\end{abstract}

\section{Introduction}



Intuitionistic fuzzy logic ($\ifl$) has attracted considerable attention due to its unique nature as a logic blending fuzzy reasoning and constructive reasoning~\cite{Avr91,BaaPreZac07,BaaZac00,Haj98,TakTit84}. The logic, which was initially defined by Takeuti and Titani in~\cite{TakTit84}, has its roots in the work of Kurt G\"odel. G\"odel introduced extensions of propositional intuitionistic logic (now called, ``G\"odel logics'') in order to prove that propositional intuitionistic logic does not possess a finite characteristic matrix~\cite{God32}. These logics were later studied by Dummett who extended G\"odel's finite-valued semantics to include an infinite number of truth-values~\cite{Dum59}. Dummett additionally provided an axiomatization for the propositional fragment of $\ifl$~\cite{Dum59}. The first-order logic $\ifl$ also admits a finite axiomatization, obtained by extending an axiomatization of first-order intuitionistic logic with the \emph{linearity axiom} $(A \imp B) \lor (B \imp A)$ and the \emph{quantifier shift axiom} $(\forall x) (A(x) \lor C) \imp \forall x A(x) \lor C$ (where $x$ does not occur free in $C$)~\cite{Hor69}.

Over the last few decades, propositional and first-order G\"odel logics (including the prominent logic $\ifl$) have been applied in various areas of logic and computer science~\cite{BaaPreZac07,BaaZac00,BorDisPen14,Haj98,LelKuz18,LifPeaVal01,Vis82}. For example, Visser~\cite{Vis82} applied the propositional fragment of $\ifl$ while analyzing the provability logic of Heyting arithmetic, Lifschitz et al.~\cite{LifPeaVal01} employed a G\"odel logic to model the strong equivalence of logic programs, and Borgwardt et al.~\cite{BorDisPen14} studied standard reasoning problems of first-order G\"odel logics in the context of fuzzy description logics. Additionally---and quite significantly---the logic $\ifl$ has been recognized as one of the fundamental formalizations of fuzzy logic~\cite{Haj98}.

The question of whether or not a logic possesses an \emph{analytic} proof calculus---that is, a calculus which stepwise (de)composes the formula to be proved---is of critical importance. Such calculi are effective tools for designing automated reasoning procedures and for proving certain (meta-)logical properties of a logic. For example, analytic calculi have been leveraged to provide decidability procedures for logics~\cite{Gen35}, to prove that logics interpolate~\cite{LelKuz18}, for counter-model extraction~\cite{LyoBer19}, and to understand the computational content of proofs~\cite{Mas93}.

In his seminal work~\cite{Gen35}, Gentzen proposed the \emph{sequent calculus} framework for classical and intuitionistic logic, and subsequently, proved his celebrated \emph{Hauptsatz} (i.e. cut-elimination theorem), which ultimately provided analytic calculi for the two logics. Gentzen's sequent calculus formalism has become one of the preferred proof-theoretic frameworks for providing analytic calculi, and indeed, many logics of interest have been equipped with such calculi. Nevertheless, one of the alluring features of the formalism---namely, its simplicity---has also proven to be one of the formalism's drawbacks; there remain many logics for which no cut-free, or analytic, sequent calculus (\`a la Gentzen) is known~\cite{GorPosTiu08,Pog08}. In response to this, the sequent calculus formalism has been extended in various ways over the last 30 years to include additional structure, allowing for numerous logics to be supplied with cut-free, analytic calculi. Some of the most prominent extensions of Gentzen's formalism include \emph{display calculi}~\cite{Bel82}, \emph{labelled calculi}~\cite{LyoBer19,Vig00}, \emph{hypersequent calculi}~\cite{Pog08}, and \emph{nested calculi}~\cite{Fit14,GorPosTiu08}.


In this paper, we employ the \emph{linear nested sequent} formalism, introduced by Lellmann in~\cite{Lel15}. Linear nested sequents fall within the \emph{nested calculus paradigm}, but where sequents are restricted to linear, instead of treelike, structures. Linear nested sequents are based off of Masini's 2-sequent framework~\cite{Mas92,Mas93} that was used to provide cut-free calculi for the modal logic $\mathsf{KD}$ as well as various other constructive logics. 
The linear nested formalism proves to be highly compatible with the well-known first-order G\"odel logic $\ifl$ (i.e. intuitionistic fuzzy logic), due to the fact that $\ifl$ can be semantically characterized by \emph{linear} relational frames (see Sect.~\ref{sec:logical-prelims}). The present work provides the linear nested calculus $\calc$ for $\ifl$, which enjoys a variety of fruitful properties, such as:\footnote{We refer to~\cite{Wan94} for a detailed discussion of fundamental proof-theoretic properties.}
\begin{itemize}

\item[$\blacktriangleright$] \emph{Separation}: Each logical rule exhibits no other logical connectives than the one to be introduced.

\item[$\blacktriangleright$] \emph{Symmetry}: Each logical connective has a left and right introduction rule.

\item[$\blacktriangleright$] \emph{Internality}: Each sequent translates into a formula of the logical language.

\item[$\blacktriangleright$] \emph{Cut-eliminability}: There exists an algorithm allowing the permutation of a $\cut$ rule (encoding reasoning with lemmata) upwards in a derivation until the rule is completely eliminated from the derivation.

\item[$\blacktriangleright$] \emph{Subformula property}: Every formula occurring in a derivation is a subformula of some formula in the end sequent.

\item[$\blacktriangleright$] \emph{Admissibility of structural rules}: Everything derivable with a structural rule (cf. $\iw$ and $\mrg$ in Sect.~\ref{sec:proof-properties}) is derivable without the structural rule. 

\item[$\blacktriangleright$] \emph{Invertibility of rules}: If the conclusion of an inference rule is derivable, then so is the premise.

\end{itemize}

In~\cite{BaaZac00}, a cut-free hypersequent calculus $\mathsf{HIF}$ for $\ifl$ was introduced to overcome the shortcomings of previously introduced systems~\cite{Hor69,TakTit84} that violated fundamental proof-theoretic properties such as cut-elimination. In contrast to $\mathsf{HIF}$, the current approach of exploiting linear nested sequents has two main benefits. First, the admissibility of structural rules has not been shown in $\mathsf{HIF}$, and as such, the calculus does not offer a purely formula-driven approach to proof search. Therefore, the calculus $\calc$ serves as a better basis for automated reasoning in $\ifl$---bottom-up applications of the rules in $\calc$ simply decompose or propagate formulae, and so, the question of if/when structural rules need to be applied does not arise. Second, the calculus $\mathsf{HIF}$ cannot be leveraged to prove interpolation for the logic $\ifl$ (see~\cite{LelKuz18}) via the so-called \emph{proof-theoretic method} (cf.~\cite{LelKuz18,LyoTiuGorClo20}) due to the presence of the communication structural rule~\cite{Avr91}. In~\cite{LelKuz18}, it was shown that the propositional fragment of $\calc$ can be harnessed to prove Lyndon interpolation for the propositional fragment of $\ifl$. This result suggests that $\calc$, in conjunction with the aforementioned proof-theoretic method, may potentially be harnessed to study and determine interpolable fragments of $\ifl$, or to solve the longstanding open problem of if the entire logic $\ifl$ interpolates or not.

The contributions and organization of this paper can be summarized as follows: In Sect.~\ref{sec:logical-prelims}, we introduce the semantics and axiomatization for intuitionistic fuzzy logic ($\ifl$). Sect.~\ref{sec:sound-complete} introduces linear nested sequents and the calculus $\calc$, as well as proves the calculus sound and complete relative to $\ifl$. In Sect.~\ref{sec:proof-properties}, we provide invertibility, structural rule admissibility, and cut-elimination results. Last, Sect.~\ref{sec:conclusion} concludes and discusses future work.

\section{Logical Preliminaries}\label{sec:logical-prelims}

Our language consists of denumerably many \emph{variables} $\{x, y, z, \ldots\}$, 
denumerably many $n$-ary \emph{predicates} $\{p, q, r, \ldots\}$ (with $n \in \mathbb{N}$), the \emph{connectives} $\bot$, $\land$, $\lor$, $\supset$, the \emph{quantifiers} $\forall$, $\exists$, and \emph{parentheses} $`($' and $`)$'. 
We define the language $\lang$ via the BNF grammar below, and will use $A, B, C,$ etc. to represent formulae from $\lang$.
\begin{center}
$A ::= p(x_{1}, \ldots, x_{n}) \ | \ \bot \ | \ (A \vee A) \ | \ (A \wedge A) \ | \ (A \imp A) \ | \ (\forall x) A \ | \ (\exists x) A $
\end{center}
In the above grammar, $p$ is any $n$-ary predicate symbol and $x_{1}, \ldots, x_{n}, x$ are variables. We refer to formulae of the form $p(x_{1}, \ldots, x_{n})$ as \emph{atomic formulae}, and (more specifically) refer to formulae of the form $p$ as \emph{propositional variables} (i.e. a $0$-ary predicate $p$ is a propositional variable). The \emph{free variables} of a formula $A$ are defined in the usual way as variables unbound by a quantifier, and \emph{bound variables} as those bounded by a quantifier.

We opt for the relational semantics of $\ifl$---as opposed to the fuzzy semantics (cf.~\cite{BaaZac00})---since the structure of linear nested sequents is well-suited for interpretation via linear relational frames.

\begin{definition}[Relational Frames, Models~\cite{GabSheSkv09}]\label{def:frame-model} A \emph{relational frame} is a triple $F = (W,R,D)$ such that: $(i)$ $W$ is a non-empty set of worlds $w, u, v, \ldots$, $(ii)$ $R$ is a reflexive, transitive, antisymmetric, and connected binary relation on $W$, 
 and $(iii)$ $D$ is a function that maps a world $w \in W$ to a non-empty set of \emph{parameters} $D_{w}$ called the \emph{domain of $w$} such that the following condition is met:
\begin{flushleft}
$\mathbf{(CD)}$ \ \ If $Rwu$, then $D_{w} = D_{u}$.
\end{flushleft}

A \emph{model} $M$ is a tuple $(F,V)$ where $F$ is a relational frame and $V$ is a \emph{valuation function} such that $V(p,w) \subseteq (D_{w})^{n}$ for each $n$-ary predicate $p$ and
\begin{flushleft}
$\mathbf{(TP)}$ \ \ If $Rwu$, then $V(p,w) \subseteq V(p,u)$ (if $p$ is of arity $n > 0$);\\
\hspace{3em} If $Rwu$ and $w \in V(p,w)$, then $u \in V(p,v)$ (if $p$ is of arity $0$).
\end{flushleft}

We uphold the convention in~\cite{GabSheSkv09} and assume that for each world $w \in W$, $(D_{w})^{0} = \{w\}$, so $V(p,w) = \{w\}$ or $V(p,w) = \emptyset$, for a propositional variable $p$.

\end{definition}

The distinctive feature of relational frames for $\ifl$ is the \emph{connected} property, which states that for any $w,u,v \in W$ of a frame $F = (W,R,D)$, if $Rwu$ and $Rwv$, then either $Ruv$ or $Rvu$. Imposing this property on reflexive, transitive, and antisymmetric (i.e. intuitionistic) frames causes the set of worlds to become linearly ordered, thus validating the linearity axiom $(A \imp B) \lor (B \imp A)$ (shown in Fig.~\ref{fig:axioms-ifl}). Furthermore, the constant domain condition (\textbf{CD}) validates the quantifier shift axiom $\forall x (A(x) \lor B) \imp \forall x A(x) \lor B$ (also shown in Fig.~\ref{fig:axioms-ifl}). 

Rather than interpret formulae from $\lang$ in relational models, we follow~\cite{GabSheSkv09} and introduce $D_{w}$-sentences to be interpreted in relational models. This gives rise to a notion of validity for formulae in $\lang$ (see Def.~\ref{def:semantic-clauses}). The definition of validity also depends on the \emph{universal closure} of a formula: if a formula $A$ contains only $x_{0}, \ldots, x_{m}$ as free variables, then the universal closure $\uc A$ is taken to be the formula $\forall x_{0} \ldots \forall x_{m} A$.

\begin{definition}[$D_{w}$-Sentence]\label{def:d-sentence} Let $M$ be a relational model with $w \in W$ of $M$. We define $\lang_{D_{w}}$ to be the language $\lang$ expanded with parameters from the set $D_{w}$. We define a \emph{$D_{w}$-formula} to be a formula in $\lang_{D_{w}}$, and we define a \emph{$D_{w}$-sentence} to be a $D_{w}$-formula that does not contain any free variables. Last, we use $a, b, c, \ldots$ to denote parameters in a set $D_{w}$.
\end{definition}

\begin{definition}[Semantic Clauses~\cite{GabSheSkv09}]
\label{def:semantic-clauses} Let $M = (W,R,D,V)$ be a relational model with $w~{\in}~W$ and $R(w) := \{v \in W \ | \ (w,v) \in R\}$. The \emph{satisfaction relation} $M,w \Vdash A$ between $w \in W$ and a $D_{w}$-sentence $A$ is inductively defined as follows:

\begin{itemize}

\item $M,w \not\Vdash \bot$

\item If $p$ is a propositional variable, then $M,w \Vdash p$ iff $w \in V(p,w)$;

\item If $p$ is an $n$-ary predicate symbol $($with $n > 0)$, then $M,w \Vdash p(a_{1}, \cdots, a_{n})$ iff $(a_{1}, \cdots, a_{n}) \in V(p,w)$;

\item $M,w \Vdash A \vee B$ iff $M,w \Vdash A$ or $M,w \Vdash B$;

\item $M,w \Vdash A \land B$ iff $M,w \Vdash A$ and $M,w \Vdash B$;

\item $M,w \Vdash A \imp B$ iff for all $u \in R(w)$, if $M,u \Vdash A$, then $M,u \Vdash B$;

\item $M,w \Vdash \forall x A(x)$ iff for all $u \in R(w)$ and all $a \in D_{u}$, $M,u \Vdash A(a)$;

\item $M,w \Vdash \exists x A(x)$ iff there exists an $a \in D_{w}$ such that $M,w \Vdash A(a)$.

\end{itemize}


We say that a formula $A$ is \emph{globally true on $M$}, written $M \Vdash A$, iff $M,u \Vdash \overline{\forall} A$ for all worlds $u \in W$. A formula $A$ is \emph{valid}, written $ \Vdash A$, iff it is globally true on all relational models.

\end{definition}

\begin{lemma}[Persistence] 
\label{lm:persistency} Let $M$ be a relational model with $w,u \in W$ of $M$. For any $D_{w}$-sentence $A$, if $M,w \Vdash A$ and $Rwu$, then $M,u \Vdash A$.

\end{lemma}

\begin{proof} See~{\cite[Lem.~3.2.16]{GabSheSkv09}} for details.
\qed
\end{proof}

A sound and complete axiomatization for the logic $\ifl$ is provided in Fig.~\ref{fig:axioms-ifl}. We define the \emph{substitution} $[y/x]$ of the variable $y$ for the free variable $x$ on a formula $A$ in the standard way as the replacement of all free occurrences of $x$ in $A$ with $y$. The substitution $[a/x]$ of the parameter $a$ for the free variable $x$ is defined similarly. Last, the side condition \emph{$y$ is free for $x$} (see~Fig.~\ref{fig:axioms-ifl}) is taken to mean that $y$ does not become bound by a quantifier if substituted for $x$.

\begin{figure}
\noindent\hrule
\[
A \supset (B \supset A)
\qquad
(A \supset (B \supset C)) \supset ((A \supset B) \supset (A \supset C))
\qquad
A \supset (B \supset (A \land B))
\]
\[
(A \land B) \supset A
\qquad
(A \land B) \supset B
\qquad
A \supset (A \lor B)
\qquad
B \supset (A \lor B)
\qquad
\infer[mp]
{B}
{A & A \imp B}
\]
\[
(A \supset C) \supset ((B \supset C) \supset ((A \lor B) \supset C))
\qquad
\bot \imp A
\qquad
(A \imp B) \lor (B \imp A)
\qquad
\infer[gen]
{\forall x A}
{A}
\]
\quad
\[
\forall x A \supset A[y/x]~\textit{y free for x}
\quad
A[y/x] \supset \exists x A~\textit{y free for x}
\quad
\forall x (B \imp A(x)) \imp (B \imp \forall x A(x))
\]
\quad
\[
\forall x (A(x) \imp B) \imp (\exists x A(x) \imp B)
\qquad
\forall x (A(x) \vee B) \imp \forall x A(x) \vee B~\textit{ with $x \not\in B$}
\]
\hrule
\caption{Axiomatization for the logic $\gl$~\cite{GabSheSkv09}. The logic $\ifl$ is the smallest set of formulae from $\lang$ closed under substitutions of the axioms and applications of the inference rules $mp$ and $gen$. We write $\vdash_{\ifl} A$ to denote that $A$ is an element, or \emph{theorem}, of $\ifl$.}
\label{fig:axioms-ifl}
\end{figure}


\begin{theorem}[Adequacy of $\gl$] For any $A \in \lang$, $\Vdash A$ iff $\vdash_{\gl} A$.

\end{theorem}

\begin{proof} The forward direction follows from~{\cite[Prop.~7.2.9]{GabSheSkv09}} and~{\cite[Prop.~7.3.6]{GabSheSkv09}}, and the backwards direction follows from~{\cite[Lem.~3.2.31]{GabSheSkv09}}.
\qed
\end{proof}

\section{Soundness and Completeness of $\calc$}\label{sec:sound-complete}

Let us define \emph{linear nested sequents} (which we will refer to as \emph{sequents}) to be syntactic objects $\g$ given by the BNF grammar shown below:
$$
\g ::= \Gamma \vdash \Gamma \ | \ \g \slash \g \textit{ where } \Gamma ::= A \ | \ \Gamma, \Gamma \textit{ with } A \in \lang.
$$
Each sequent $\g$ is of the form $\Gamma_{1} \vdash \Delta_{1} \slash \cdots \slash \Gamma_{n} \vdash \Delta_{n}$ with $n \in \mathbb{N}$. We refer to each $\Gamma_{i} \vdash \Delta_{i}$ (for $1 \leq i \leq n$) as 
a \emph{component} of $\g$ and use $|| \g ||$ to denote the number of components in $\g$. 

We often use $\g$, $\h$, $\f$, and $\k$ to denote sequents, and will use $\Gamma$ and $\Delta$ to denote antecedents and consequents of components. Last, we take the comma operator to be commutative and associative; for example, we identify the sequent $p(x) \vdash q(x), r(y), p(x)$ with $p(x) \vdash r(y), p(x), q(x)$. This interpretation lets us view an antecedent $\Gamma$ or consequent $\Delta$ as a multiset of formulae.

To ease the proof of cut-elimination (Thm.~\ref{thm:cut-elimination}), we follow~\cite{Fit14} and syntactically distinguish between \emph{bound variables} $\{x, y, z, \ldots\}$ and \emph{parameters} $\{\fvx, \fvy, \fvz, \ldots\},$ which will take the place of free variables occurring in formulae. Thus, our sequents make use of formulae from $\lang$ where each free variable has been replaced by a unique parameter. For example, we would use the sequent $p(a) \vdash \forall x q(x,b) \slash \bot \vdash r(a)$ instead of the sequent $p(x) \vdash \forall x q(x,y) \slash \bot \vdash r(x)$ in a derivation (where the parameter $a$ has been substituted for the free variable $x$ and $b$ has been substituted for $y$). We also use the notation $A(a_{0}, \ldots, a_{n})$ to denote that the parameters $a_{0}, \ldots, a_{n}$ occur in the formula $A$, and write $A(\vv{a})$ as shorthand for $A(a_{0}, \ldots, a_{n})$. This notation extends straightforwardly to sequents as well.

The linear nested calculus $\calc$ for $\ifl$ is given in Fig.~\ref{fig:lngl}. (NB. The linear nested calculus $\lng$ introduced in~\cite{LelKuz18} is the propositional fragment of $\calc$, i.e. $\lng$ is the calculus $\calc$ without the quantifier rules and where propositional variables are used in place of atomic formulae.) The $\imprtwo$ and $\allrtwo$ rules in $\calc$ are particularly noteworthy; as will be seen in the next section, the rules play a vital role in ensuring the invertibility and admissibility of certain rules, ultimately permitting the elimination of $\cut$ (see Thm.~\ref{thm:cut-elimination}).

\begin{figure}[t]
\noindent\rule{38em}{0.4pt}
\[
\infer[\idone]
{\seq{\g}{\Gamma,p(\vv{\fvx})}{p(\vv{\fvx}),\Delta}{\h}}
{}
\quad
\infer[\idtwo]
{\seq{\g}{\Gamma_{1}, p(\vv{\fvx}) \vdash \Delta_{1} \slash \h \slash \Gamma_{2}}{p(\vv{\fvx}),\Delta_{2}}{\f}}
{}
\]
\[
\infer[\botl]
{\seq{\g}{\Gamma, \bot}{\Delta}{\h}}
{}
\qquad
\infer[\landl]
{\seq{\g}{ \Gamma, A \wedge B}{\Delta}{\h}}
{\seq{\g}{ \Gamma, A, B}{\Delta}{\h}}
\qquad
\infer[\lorr]
{\seq{\g}{\Gamma}{\Delta, A \lor B}{\h}}
{\seq{\g}{\Gamma}{\Delta, A, B}{\h}}
\]
\[
\infer[\landr]
{\seq{\g}{ \Gamma}{\Delta, A \wedge B}{\h}}
{\seq{\g}{ \Gamma}{\Delta, A}{\h} & \seq{\g}{ \Gamma}{\Delta, B}{\h}}
\quad
\infer[\lorl]
{\seq{\g}{\Gamma, A \lor B}{\Delta}{\h}}
{\seq{\g}{\Gamma, A}{\Delta}{\h} & \seq{\g}{ \Gamma, B}{\Delta}{\h}}
\]
\[
\infer[\existsr] 
{\seq{\g}{\Gamma}{\exists x A, \Delta}{\h}}
{\seq{\g}{\Gamma}{A[\fvx/x], \exists x A, \Delta}{\h}}
\qquad
\infer[\impl]
{\seq{\g}{\Gamma, A \supset B}{\Delta}{\h}}
{\seq{\g}{ \Gamma, B}{\Delta}{\h} & \seq{\g}{\Gamma, A \supset B}{A,\Delta}{\h}}
\]
\[
\infer[\lift]
{\seq{\g}{\Gamma_{1}, A}{\Delta_{1} \sslash \h \sslash \Gamma_{2} \vdash \Delta_{2}}{\f}}
{\seq{\g}{ \Gamma_{1}, A}{\Delta_{1} \sslash \h \sslash \Gamma_{2}, A \vdash \Delta_{2}}{\f}}
\qquad
\infer[\alll] 
{\seq{\g}{\Gamma, \forall x A}{\Delta}{\h}}
{\seq{\g}{\Gamma, A[\fvx/x], \forall x A}{\Delta}{\h}}
\]
\[
\infer[\allrone^{\dag}]
{\g \sslash \Gamma \vdash \Delta, \forall x A}
{\seq{\g}{\Gamma}{\Delta}{ \vdash A[\fvx/x] }}
\
\infer[\existsl^{\dag}]
{\seq{\g}{\Gamma, \exists x A}{\Delta}{\h}}
{\seq{\g}{\Gamma, A[\fvx/x]}{\Delta}{\h}}
\
\infer[\improne]
{\g \sslash \Gamma \vdash \Delta, A \supset B}
{\g \sslash \Gamma \vdash \Delta \sslash A \vdash B}
\]
\[
\infer[\imprtwo]
{\g \sslash \Gamma_{1} \vdash \Delta_{1}, A \supset B \sslash  \Gamma_{2} \vdash \Delta_{2} \sslash \h}
{{\g \sslash \Gamma_{1} \vdash \Delta_{1} \sslash A \vdash B \sslash \Gamma_{2} \vdash \Delta_{2} \sslash \h} & {\g \sslash \Gamma_{1} \vdash \Delta_{1} \sslash \Gamma_{2} \vdash \Delta_{2}, A \supset B \sslash \h}}
\]
\[
\infer[\allrtwo^{\dag}]
{\g \sslash \Gamma_{1} \vdash \Delta_{1}, \forall x A \sslash \Gamma_{2} \vdash \Delta_{2} \sslash \h}
{{\g \sslash \Gamma_{1} \vdash \Delta_{1} \sslash \vdash A[\fvx/x] \sslash \Gamma_{2} \vdash \Delta_{2} \sslash \h} & {\g \sslash \Gamma_{1} \vdash \Delta_{1} \sslash \Gamma_{2} \vdash \Delta_{2}, \forall x A \sslash \h}}
\]
\noindent\rule{38em}{0.4pt}
\caption{The Calculus $\calc$. The side condition $\dag$ stipulates that the parameter $\fvx$ is an \emph{eigenvariable}, i.e. it does not occur in the conclusion. Occasionally, we write $\vdash_{\calc} \g$ to mean that the sequent $\g$ is derivable in $\calc$.} 
\label{fig:lngl}
\end{figure}

To obtain soundness, we interpret each sequent as a formula in $\lang$ and utilize the notion of validity in Def.~\ref{def:semantic-clauses}. The following definition specifies how each sequent is interpreted.

\begin{definition}[Interpretation $\iota$]
\label{def:interpretation-of-lns}
The interpretation of a sequent is defined inductively as follows:

\begin{center}
\begin{tabular}{c @{\hskip 3em} c}

$\iota(\Gamma \vdash \Delta) := \displaystyle{\bigwedge \Gamma \imp \bigvee \Delta}$

&

$\iota(\Gamma \vdash \Delta \sslash \g) := \displaystyle{ \bigwedge \Gamma \imp \bigg ( \bigvee \Delta \vee \iota(\g) \bigg )}$

\end{tabular}
\end{center}
We interpret a sequent $\g$ as a formula in $\lang$ by taking the universal closure $\uc \iota(\g)$ of $\iota(\g)$ and we say that $\g$ is valid if and only if $\Vdash \uc \iota(\g)$. 

\end{definition}

\begin{theorem}[Soundness of $\calc$]\label{thm:soundness-calc}
For any linear nested sequent $\g$, if $\g$ is provable in $\calc$, then $\Vdash \uc \iota(\g)$.
\end{theorem}

\begin{proof} We prove the result by induction on the height of the derivation of
$$
\g = \Gamma_{1} \vdash \Delta_{1} \sslash \cdots \sslash \Gamma_{n} \vdash \Delta_{n} \sslash \Gamma_{n+1} \vdash \Delta_{n+1} \sslash \cdots \sslash \Gamma_{m} \vdash \Delta_{m}
$$
and only present the more interesting $\forall$ quantifier cases in the inductive step. All remaining cases can be found in App.~\ref{appendix}. Each inference rule considered is of one of the following two forms. 
\begin{center}
\begin{tabular}{c @{\hskip 5em} c}
\AxiomC{$\g'$}
\RightLabel{$(r_{1})$}
\UnaryInfC{$\g$}
\DisplayProof

&

\AxiomC{$\g_{1}$}
\AxiomC{$\g_{2}$}
\RightLabel{$(r_{2})$}
\BinaryInfC{$\g$}
\DisplayProof
\end{tabular}
\end{center}

We argue by contraposition and prove that if $\g$ is invalid, then at least one premise is invalid. Assuming $\g$ is invalid (i.e. $\not\Vdash \uc \iota(\g)$) implies the existence of a model $M = (W,R,D,V)$ with world $v \in W$ such that $Rvw_{0}$, $\vv{a} \in D_{w_{0}}$, and $M, w_{0} \not\Vdash \iota(\g)(\vv{a})$, where $\vv{a}$ represents all parameters in $\iota(\g)$. Hence, there is a sequence of worlds $w_{1}, \cdots, w_{m} \in W$ such that $Rw_{j}w_{j+1}$ (for $0 \leq j \leq m-1$), $M, w_{i} \Vdash \bigwedge \Gamma_{i}$, and $M, w_{i} \not\Vdash \bigvee \Delta_{i}$, for each $1 \leq i \leq m$. We assume all parameters in $\bigwedge \Gamma_{i}$ and $\bigvee \Delta_{i}$ are interpreted as elements of the associated domain $D_{w_{i}}$ (for $1 \leq i \leq m$). 

{\bf $\allrone$-rule:} By our assumption $M, w_{m} \Vdash \bigwedge \Gamma_{m}$ and $M, w_{m} \not\Vdash \bigvee \Delta_{m} \vee \forall x A$. The latter implies that $M, w_{m} \not\Vdash \forall x A$, meaning there exists a world $w_{m+1} \in W$ such that $Rw_{m}w_{m+1}$ and $M, w_{m+1} \not\Vdash A[b/x]$ for some $b \in D_{w_{m+1}}$. If we interpret the eigenvariable of the premise as $b$, then the premise is shown invalid.


{\bf $\allrtwo$-rule:} It follows from our assumption that $M, w_{n} \Vdash \bigwedge \Gamma_{n}$, $M, w_{n} \not\Vdash \bigvee \Delta_{n} \vee \forall x A$, $M, w_{n+1} \Vdash \bigwedge \Gamma_{n+1}$, and $M, w_{n+1} \not\Vdash \bigvee \Delta_{n+1}$. The fact that $M, w_{n} \not\Vdash \bigvee \Delta_{n} \vee \forall x A$ implies that there exists a world $w \in W$ such that $Rw_{n}w$ and for some $b \in D_{w}$, $M, w \not\Vdash A[b / x]$. Since our frames are connected, there are two cases to consider: (i) $ R w w_{n+1}$, or (ii) $R w_{n+1} w$. Case (i) falsifies the left premise, and case (ii) falsifies the right premise.

{\bf $\alll$-rule:} We know that $M, w_{n} \Vdash \bigwedge \Gamma_{n} \wedge \forall x A$ and $M, w_{n} \not\Vdash \bigvee \Delta_{n}$. Hence, for any world $w \in W$, if $Rw_{n}w$, then $M, w \Vdash A[b/x]$ for all $b \in D_{w}$. Since $Rw_{n}w_{n}$, it follows that $M, w_{n} \Vdash A[\fvy / x]$ for any $b \in D_{w_{n}}$. If $a$ occurs in the conclusion $\g$, then by the constant domain condition (\textbf{CD}), we know that $a \in D_{w_{n}}$, so we may falsify the premise of the rule. If $a$ does not occur in $\g$, then it is an eigenvariable, and assigning $a$ to any element of $D_{w_{n}}$ will falsify the premise. 
\qed
\end{proof}

\begin{theorem}[Completeness of $\calc$]\label{thm:completeness-calc} If $\vdash_{\ifl} A$, then $A$ is provable in $\calc$. 

\end{theorem}

\begin{proof} It is not difficult to show that $\calc$ can derive each axiom of $\ifl$ and can simulate each inference rule. We refer the reader to App.~\ref{appendix} for details.
\qed
\end{proof}

\section{Proof-Theoretic Properties of $\calc$}\label{sec:proof-properties}

In this section, we present the fundamental proof-theoretic properties of $\calc$, thus extending the results in~\cite{LelKuz18} from the propositional setting to the first-order setting. (NB. We often leverage results from~\cite{LelKuz18} to simplify our proofs.) 
With the exception of Lem.~\ref{lm:admiss-icl}, Lem.~\ref{lm:admiss-icr}, and Thm.~\ref{thm:cut-elimination}, all results are proved by induction on the \emph{height} of a given derivation $\Pi$, i.e. on the length (number of sequents) of the longest branch from the end sequent to an initial sequent in $\Pi$. %
Lemmata whose proofs are omitted can be found in App.~\ref{appendix}.

\begin{figure}
\hrule
\[
\infer[\iw]
{\seq{\g}{\Gamma_{1}, \Gamma_{2}}{\Delta_{1},\Delta_{2}}{\h}}
{\seq{\g}{\Gamma_{1}}{\Delta_{1}}{\h}}
\qquad
\infer[\icl]
{\seq{\g}{\Gamma, A}{\Delta}{\h}}
{\seq{\g}{\Gamma,A,A}{\Delta}{\h}}
\qquad
\infer[\ew]
{\g \sslash \ \vdash \ \sslash \h}
{\g \sslash \h}
\]
\[
\infer[\vsub]
{\g[\fvy / \fvx]}
{\g}
\qquad
\infer[\lwr]
{\g \sslash \Gamma_{1} \vdash \Delta_{1} \sslash \Gamma_{2} \vdash A, \Delta_{2} \sslash \h }
{\g \sslash \Gamma_{1} \vdash A, \Delta_{1} \sslash \Gamma_{2} \vdash \Delta_{2} \sslash \h }
\qquad
\infer[(\bot_{r})]
{\g \slash \Gamma \vdash \Delta \slash \h}
{\g \slash \Gamma \vdash \Delta, \bot \slash \h}
\]
\[
\infer[\icr]
{\seq{\g}{\Gamma, \Gamma}{A,\Delta}{\h}}
{\seq{\g}{\Gamma}{A,A,\Delta}{\h}}
\qquad
\infer[\mrg]
{\g \sslash \Gamma_{1}, \Gamma_{2} \vdash \Delta_{1}, \Delta_{2} \sslash \h }
{\g \sslash \Gamma_{1} \vdash \Delta_{1} \sslash \Gamma_{2} \vdash \Delta_{2} \sslash \h }
\]
\noindent\rule{38em}{0.4pt}
\caption{Admissible rules in $\calc$.}
\label{fig:admissible-structural-rules}
\end{figure}

We say that a rule is \emph{admissible} in $\calc$ iff derivability of the premise(s) implies derivability of the conclusion in $\calc$. Additionally, a rule is \emph{height preserving (hp-)admissible} in $\calc$ iff if the premise of the rule has a derivation of a certain height in $\calc$, then the conclusion of the rule has a derivation of the same height or less in $\calc$. Last, a rule is \emph{invertible} (\emph{hp-invertible}) iff derivability of the conclusion implies derivability of the premise(s) (with a derivation of the same height or less). Admissible rules of $\calc$ are given in Fig.~\ref{fig:admissible-structural-rules}.

\begin{lemma}\label{lm:A-implies-A} For any $A$, $\Gamma$, $\Delta$, $\g$, and $\h$, $\vdash_{\calc} \seq{\g}{\Gamma, A}{A, \Delta}{\h}$.

\end{lemma}

\begin{lemma}\label{lem:botr-admiss} The $(\bot_{r})$ rule is hp-admissible in $\calc$.

\end{lemma}

\begin{proof} By induction on the height of the given derivation. In the base case, applying $(\bot_{r})$ to $\idone$, $\idtwo$, or $\botl$ gives an initial sequent, and for each case of the inductive step we apply IH followed by the corresponding rule.
\qed
\end{proof}

\begin{lemma}\label{lm:substitution-lemma}
The $\vsub$ rule is hp-admissible in $\calc$.
\end{lemma}

\begin{lemma}
\label{lm:hp-admiss-iw}
The $\iw$ rule is hp-admissible in $\calc$.
\end{lemma}

\begin{lemma}\label{lm:admiss-ew}
The $\ew$ rule is admissible in $\calc$.
\end{lemma}

\begin{proof} By~\cite[Lem.~5.6]{LelKuz18} we know that $\ew$ is admissible in $\lng$, thus leaving us to prove the $\allrone$, $\existsl$, $\allrtwo$, $\alll$, and $\existsr$ cases. The $\existsl$, $\alll$, and $\existsr$ cases are easily shown by applying IH and then the rule. We therefore prove the $\allrone$ and $\allrtwo$ cases, beginning with the former, which is split into the two subcases, shown below:

\begin{center}
\begin{tabular}{c @{\hskip 2em} c}

\AxiomC{$\g \slash \Gamma \vdash \Delta \slash \vdash A[\fvx/x]$}
\RightLabel{$\allrone$}
\UnaryInfC{$\g \slash \Gamma \vdash \Delta, \forall x A$}
\RightLabel{$\ew$}
\UnaryInfC{$\g' \slash \Gamma \vdash \Delta, \forall x A$}
\DisplayProof

&

\AxiomC{$\g \slash \Gamma \vdash \Delta \slash \vdash A[\fvx/x]$}
\RightLabel{$\allrone$}
\UnaryInfC{$\g \slash \Gamma \vdash \Delta, \forall x A$}
\RightLabel{$\ew$}
\UnaryInfC{$\g \slash \Gamma \vdash \Delta, \forall x A \slash \vdash$}
\DisplayProof

\end{tabular}
\end{center}

In the top left case, where we weaken in a component prior to the component $\Gamma \vdash \Delta, \forall x A$, we may freely permute the two rule instances. The top right case is resolved as shown below.
\begin{center}
\begin{tabular}{c}

\AxiomC{}
\RightLabel{IH}
\dashedLine
\UnaryInfC{$\g \slash \Gamma \vdash \Delta \slash \vdash A[\fvx/x] \slash \vdash$}

\AxiomC{}
\RightLabel{IH}
\dashedLine
\UnaryInfC{$\g \slash \Gamma \vdash \Delta \slash \vdash \slash \vdash A[\fvx/x]$}

\RightLabel{$\allrone$}

\UnaryInfC{$\g \slash \Gamma \vdash \Delta \slash \vdash \forall x A$}
\RightLabel{$\allrtwo$}
\BinaryInfC{$\g \slash \Gamma \vdash \Delta, \forall x A \slash \vdash$}
\DisplayProof

\end{tabular}
\end{center}

Suppose now that we have an $\allrtwo$ inference (as in Fig.~\ref{fig:lngl}) followed by an $\ew$ inference. 
The only nontrivial case (which is resolved as shown below) occurs when a component is weakened in directly after the component $ \Gamma_{1} \vdash \Delta_{1}, \forall x A$. All other cases follow by an application of IH followed by an application of the $\allrtwo$ rule.
\begin{center}
\begin{tabular}{c}
\AxiomC{}
\RightLabel{IH}
\dashedLine
\UnaryInfC{$\g \slash \Gamma_{1} \vdash \Delta_{1} \slash \vdash A[\fvx/x] \slash \vdash \slash \Gamma_{2} \vdash \Delta_{2} \slash \h $}

\AxiomC{$\Pi$}
\RightLabel{$\allrtwo$}
\BinaryInfC{$\g \slash \Gamma_{1} \vdash \Delta_{1}, \forall x A \slash \vdash \slash \Gamma_{2} \vdash \Delta_{2} \slash \h $}
\DisplayProof
\end{tabular}
\end{center}
\begin{center}
\resizebox{\columnwidth}{!}{
\begin{tabular}{c c c} 
$\Pi$

&

$= \Bigg \{$ 

&

\AxiomC{}
\RightLabel{IH}
\dashedLine
\UnaryInfC{$\g \slash \Gamma_{1} \vdash \Delta_{1} \slash \vdash \slash \vdash A[\fvx/x] \slash \Gamma_{2} \vdash \Delta_{2} \slash \h $}

\AxiomC{}
\RightLabel{IH}
\dashedLine
\UnaryInfC{$\g \slash \Gamma_{1} \vdash \Delta_{1} \slash \vdash \slash \Gamma_{2} \vdash \Delta_{2}, \forall x A \slash \h $}
\RightLabel{$\allrtwo$}
\BinaryInfC{$\g \slash \Gamma_{1} \vdash \Delta_{1} \slash  \vdash \forall x A \slash \Gamma_{2} \vdash \Delta_{2} \slash \h $}
\DisplayProof
\end{tabular}
}
\end{center}
\qed
\end{proof}

\begin{lemma}\label{lm:admiss-lower} The rule $\lwr$ is hp-admissible in $\calc$.
\end{lemma}

\begin{proof} By~\cite[Lem.~5.7]{LelKuz18} we know that $\lwr$ is admissible in $\lng$, and so, we may prove the claim by extending it to include the quantifier rules. We have two cases to consider: either (i) the lower-formula is a side formula in the quantifier inference, or (ii) it is principal. In case (i), the $\allrone$, $\alll$, $\existsl$, and $\existsr$ cases can be resolved by applying IH followed by an application of the rule. Concerning the $\allrtwo$ rule, all cases follow by applying IH and then the rule, with the exception of the following:
$$
\scalebox{.9}{
\infer[\lwr]
{\g \sslash \Gamma_{1} \vdash \Delta_{1}, \forall x B \sslash \Gamma_{2} \vdash \Delta_{2}, A \sslash \h}
{
\infer[\allrtwo]
{\g \sslash \Gamma_{1} \vdash \Delta_{1}, A, \forall x B \sslash \Gamma_{2} \vdash \Delta_{2} \sslash \h}
{{\g \sslash \Gamma_{1} \vdash \Delta_{1}, A \sslash \vdash B[\fvx/x] \sslash \Gamma_{2} \vdash \Delta_{2} \sslash \h} & {\g \sslash \Gamma_{1} \vdash \Delta_{1}, A \sslash \Gamma_{2} \vdash \Delta_{2}, \forall x B \sslash \h}}
}
}
$$
In the above case, we apply IH twice to the top left premise and apply IH once to the top right premise. A single application of $\allrtwo$ gives the desired result.

Let us now consider case (ii). Observe that the principal formulae in $\allrone$, $\alll$, and $\existsl$ cannot be principal in the use of $\lwr$, so we need only consider the $\existsr$ and $\allrtwo$ cases. The $\existsr$ case is shown below top-left and the case is resolved as shown below top-right. In the $\allrtwo$ case (shown below bottom), we take the derivation of the top right premise as the proof of the desired conclusion.
\begin{center}
\resizebox{\columnwidth}{!}{
\begin{tabular}{c @{\hskip 1em} c}
\AxiomC{$\g \slash \Gamma_{1} \vdash \Delta_{1}, A[\fvx/x], \exists x A \slash \Gamma_{2} \vdash \Delta_{2} \slash \h$}
\RightLabel{$\existsr$}
\UnaryInfC{$\g \slash \Gamma_{1} \vdash \Delta_{1}, \exists x A \slash \Gamma_{2} \vdash \Delta_{2} \slash \h$}
\RightLabel{$\lwr$}
\UnaryInfC{$\g \slash \Gamma_{1} \vdash \Delta_{1} \slash \Gamma_{2} \vdash \Delta_{2}, \exists x A \slash \h$}
\DisplayProof

&

\AxiomC{}
\RightLabel{IH $\times 2$}
\dashedLine
\UnaryInfC{$\g \slash \Gamma_{1} \vdash \Delta_{1} \slash \Gamma_{2} \vdash \Delta_{2}, A[\fvx/x], \exists x A \slash \h$}
\RightLabel{$\existsr$}
\UnaryInfC{$\g \slash \Gamma_{1} \vdash \Delta_{1} \slash \Gamma_{2} \vdash \Delta_{2}, \exists x A \slash \h$}
\DisplayProof
\end{tabular}
}
\end{center}
\begin{center}
\resizebox{\columnwidth}{!}{
\begin{tabular}{c}
\AxiomC{$\g \slash \Gamma_{1} \vdash \Delta_{1} \slash \vdash A[\fvx/x] \slash \Gamma_{2} \vdash \Delta_{2} \slash \h$}
\AxiomC{$\g \slash \Gamma_{1} \vdash \Delta_{1} \slash \Gamma_{2} \vdash \Delta_{2}, \forall x A  \slash \h$}

\RightLabel{$\allrtwo$}
\BinaryInfC{$\g \slash \Gamma_{1} \vdash \Delta_{1}, \forall x A \slash \Gamma_{2} \vdash \Delta_{2} \slash \h$}
\RightLabel{$\lwr$}
\UnaryInfC{$\g \slash \Gamma_{1} \vdash \Delta_{1} \slash \Gamma_{2} \vdash \Delta_{2}, \forall x A \slash \h$}
\DisplayProof
\end{tabular}
}
\end{center}
\qed
\end{proof}

Our version of the $\lift$ rule necessitates a stronger form of invertibility, called \emph{m-invertibility}, for the $\landl$, $\lorl$, $\impl$, $\alll$, and $\existsl$ rules (cf.~\cite{LelKuz18}). We use $A^{k_{i}}$ to represent $k_{i}$ copies of a formula $A$, with $i \in \mathbb{N}$.

\begin{lemma}\label{lm:m-invertibility} If $\sum_{i=1}^{n} k_{n} \geq 1$, then
\vspace*{-1 em}
\begin{center}
\begin{tabular}{c @{\hskip 1em} c @{\hskip 1em} c}
$
\begin{aligned}
& (i) \ (1) \text{ implies } (2)\\
& (ii) \ (3) \text{ implies } (4) \text{ and } (5)
\end{aligned}
$

&

$
\begin{aligned}
& (iii) \ (6) \text{ implies } (7) \text{ and } (8)\\
& (iv) \ (9) \text{ implies } (10)
\end{aligned}
$

&

$
\begin{aligned}
& (v) \ (11) \text{ implies } (12)\\
\
\end{aligned}
$
\end{tabular}
\end{center}
\vspace*{-1 em}
\begin{eqnarray}
\vdash_{\calc} \ \Gamma_{1}, (A \land B)^{k_{1}} \vdash \Delta_{1} \slash & \cdots & \slash \Gamma_{n}, (A \land B)^{k_{n}} \vdash \Delta_{n}\\
\vdash_{\calc} \ \Gamma_{1}, A^{k_{1}}, B^{k_{1}} \vdash \Delta_{1} \slash & \cdots & \slash \Gamma_{n}, A^{k_{n}}, B^{k_{n}} \vdash \Delta_{n}\\
\vdash_{\calc} \ \Gamma_{1}, (A \lor B)^{k_{1}} \vdash \Delta_{1} \slash & \cdots & \slash \Gamma_{n}, (A \lor B)^{k_{n}} \vdash \Delta_{n}\\
\vdash_{\calc} \ \Gamma_{1}, A^{k_{1}} \vdash \Delta_{1} \slash & \cdots & \slash \Gamma_{n}, A^{k_{n}} \vdash \Delta_{n}\\
\vdash_{\calc} \ \Gamma_{1}, B^{k_{1}} \vdash \Delta_{1} \slash & \cdots & \slash \Gamma_{n}, B^{k_{n}} \vdash \Delta_{n}\\
\vdash_{\calc} \ \Gamma_{1}, (A \imp B)^{k_{1}} \vdash \Delta_{1} \slash & \cdots & \slash \Gamma_{n}, (A \imp B)^{k_{n}} \vdash \Delta_{n}\\
\vdash_{\calc} \ \Gamma_{1}, B^{k_{1}} \vdash \Delta_{1} \slash & \cdots & \slash \Gamma_{n}, B^{k_{n}} \vdash \Delta_{n}\\
\vdash_{\calc} \ \Gamma_{1}, (A \imp B)^{k_{1}} \vdash \Delta_{1}, A^{k_{1}} \slash & \cdots & \slash \Gamma_{n}, (A \imp B)^{k_{n}} \vdash \Delta_{n}, A^{k_{n}}\\
\vdash_{\calc} \ \Gamma_{1}, (\forall x A)^{k_{1}} \vdash \Delta_{1} \slash & \cdots & \slash \Gamma_{n}, (\forall x A)^{k_{n}} \vdash \Delta_{n}\\
\vdash_{\calc} \ \Gamma_{1}, A[\fvx /x]^{k_{1}}, (\forall x A)^{k_{1}} \vdash \Delta_{1} \slash & \cdots & \slash \Gamma_{n}, A[\fvx /x]^{k_{n}}, (\forall x A)^{k_{n}} \vdash \Delta_{n}\\
\vdash_{\calc} \ \Gamma_{1}, (\exists x A)^{k_{1}} \vdash \Delta_{1} \slash & \cdots & \slash \Gamma_{n}, (\exists x A)^{k_{n}} \vdash \Delta_{n}\\
\vdash_{\calc} \ \Gamma_{1}, A[\fvx/x]^{k_{1}} \vdash \Delta_{1} \slash & \cdots & \slash \Gamma_{n}, A[\fvx/x]^{k_{n}}\vdash \Delta_{n}
\end{eqnarray}

\end{lemma}

\begin{lemma}\label{lm:invertibility-of-andr-orr-existsr} The $\landr$, $\lorr$, and $\existsr$ rules are hp-invertible in $\calc$.

\end{lemma}

\begin{proof} By~\cite[Lem.~5.8]{LelKuz18} we know that the claim holds for the $\landr$ and $\lorr$ rules relative to $\lng$. The proof may be extended to $\calc$ by considering the quantifier rules in the inductive step; however, it is quick to verify the claim for the quantifier rules by applying IH and then the corresponding rule. Proving invertibility of the $\existsr$ rule is straightforward, and follows from the hp-admissibility of $\iw$ (Lem.~\ref{lm:hp-admiss-iw}).
\qed
\end{proof}

\begin{lemma}\label{lm:invert-implies-right-two} The $\imprtwo$ rule is invertible in $\calc$.

\end{lemma}

\begin{proof} We extend the proof of \cite[Lem.~5.10]{LelKuz18} to include the quantifier rules, and prove the result by induction on the height of the given derivation of $\g \slash \Gamma_{1} \vdash \Delta_{1}, A \imp B \slash \Gamma_{2} \vdash \Delta_{2} \slash \h$. Derivability of the right premise $\g \slash \Gamma_{1} \vdash \Delta_{1} \slash \Gamma_{2} \vdash \Delta_{2}, A \imp B \slash \h$ follows from Lem.~\ref{lm:admiss-lower}, so we focus on showing that the left premise $\g \slash \Gamma_{1} \vdash \Delta_{1} \slash A \vdash B \slash \Gamma_{2} \vdash \Delta_{2} \slash \h$ is derivable. For the $\allrone$, $\alll$, $\existsl$, and $\existsr$ rules the desired conclusion is obtained by applying IH, followed by an application of the corresponding rule. The nontrivial $\allrtwo$ case is shown below top and is resolved as shown below bottom. In all other $\allrtwo$ cases, we apply IH followed by the $\allrtwo$ rule.
\begin{center}
\resizebox{\columnwidth}{!}{
\begin{tabular}{c c c} 
\AxiomC{$\g \slash \Gamma_{1} \vdash \Delta_{1}, A \imp B \slash \vdash C[\fvx / x] \slash \Gamma_{2} \vdash \Delta_{2} \slash \h$}
\AxiomC{$\g \slash \Gamma_{1} \vdash \Delta_{1}, A \imp B \slash \Gamma_{2} \vdash \Delta_{2}, \forall x C \slash \h$}
\RightLabel{$\allrtwo$}
\BinaryInfC{$\g \slash \Gamma_{1} \vdash \Delta_{1}, \forall x C, A \imp B \slash \Gamma_{2} \vdash \Delta_{2} \slash \h$}
\DisplayProof
\end{tabular}
}
\end{center}
\begin{center}
\resizebox{\columnwidth}{!}{
\begin{tabular}{c c c}
\AxiomC{$\Pi_{1}$}
\AxiomC{$\Pi_{2}$}
\RightLabel{$\allrtwo$}
\BinaryInfC{$\g \slash \Gamma_{1} \vdash \Delta_{1}, \forall x C \slash A \vdash B \slash \Gamma_{2} \vdash \Delta_{2} \slash \h$}
\DisplayProof

&

&

\begin{tabular}{c c c}
$\Pi_{1}$

&

$= \Bigg \{$

&

\AxiomC{$\g \slash \Gamma_{1} \vdash \Delta_{1}, A \imp B \slash \vdash C[\fvx / x] \slash \Gamma_{2} \vdash \Delta_{2} \slash \h$}
\RightLabel{Lem.~\ref{lm:admiss-lower}}
\dashedLine
\UnaryInfC{$\g \slash \Gamma_{1} \vdash \Delta_{1} \slash \vdash C[\fvx / x], A \imp B \slash \Gamma_{2} \vdash \Delta_{2} \slash \h$}
\RightLabel{IH}
\dashedLine
\UnaryInfC{$\g \slash \Gamma_{1} \vdash \Delta_{1} \slash \vdash C[\fvx / x] \slash A \vdash B \slash \Gamma_{2} \vdash \Delta_{2} \slash \h$}
\DisplayProof
\end{tabular}

\end{tabular}
}
\end{center}

\begin{center}
\resizebox{\columnwidth}{!}{
\begin{tabular}{c c c} 
$\Pi_{2}$

&

$= \Bigg \{$

&

\AxiomC{$\g \slash \Gamma_{1} \vdash \Delta_{1}, A \imp B \slash \vdash C[\fvx / x] \slash \Gamma_{2} \vdash \Delta_{2} \slash \h$}
\RightLabel{IH}
\dashedLine
\UnaryInfC{$\g \slash \Gamma_{1} \vdash \Delta_{1} \slash A \vdash B \slash \vdash C[\fvx / x] \slash \Gamma_{2} \vdash \Delta_{2} \slash \h$}

\AxiomC{$\g \slash \Gamma_{1} \vdash \Delta_{1}, A \imp B \slash \Gamma_{2} \vdash \Delta_{2}, \forall x C \slash \h$}
\RightLabel{IH}
\dashedLine
\UnaryInfC{$\g \slash \Gamma_{1} \vdash \Delta_{1} \slash A \vdash B \slash \Gamma_{2} \vdash \Delta_{2}, \forall x C \slash \h$}

\RightLabel{$\allrtwo$}
\BinaryInfC{$\g \slash \Gamma_{1} \vdash \Delta_{1} \slash A \vdash B, \forall x C \slash \Gamma_{2} \vdash \Delta_{2} \slash \h$}
\DisplayProof
\end{tabular}
}
\end{center}
\qed
\end{proof}


\begin{lemma}\label{lm:invert-all-right-two} The $\allrtwo$ rule is invertible in $\calc$.

\end{lemma}

\begin{proof} Let the sequent $\g \slash \Gamma_{1} \vdash \Delta_{1}, \forall x A \slash \Gamma_{2} \vdash \Delta_{2} \slash \h$ be derivable in $\calc$. Derivability of the right premise $\g \slash \Gamma_{1} \vdash \Delta_{1}\slash \Gamma_{2} \vdash \Delta_{2}, \forall x A \slash \h$ follows from the hp-admissibility of $\lwr$ (Lem.~\ref{lm:admiss-lower}). We prove that the left premise $\g \slash \Gamma_{1} \vdash \Delta_{1} \slash \vdash A[\fvx/x] \slash \Gamma_{2} \vdash \Delta_{2} \slash \h$ is derivable by induction on the height of the given derivation.

\textit{Base case.} Regardless of if $\g \slash \Gamma_{1} \vdash \Delta_{1}, \forall x A \slash \Gamma_{2} \vdash \Delta_{2} \slash \h$ is derived by an application of $\idone$, $\idtwo$, or $\botl$, $\g \slash \Gamma_{1} \vdash \Delta_{1} \slash \vdash A[\fvx/x] \slash \Gamma_{2} \vdash \Delta_{2} \slash \h$ is an initial sequent as well.

\textit{Inductive step.} For all rules, with the exception of $\lift$, $\imprtwo$, $\allrone$, $\existsl$, and $\allrtwo$, we apply IH to the premise(s) followed by the corresponding rule. We consider the aforementioned nontrivial cases below.

If the $\lift$ rule is applied as shown below left, then the desired conclusion may be derived as shown below right. In all other cases, we apply IH and then $\lift$ to achieve the desired result. 
\begin{center}
\resizebox{\columnwidth}{!}{
\begin{tabular}{c c} 
\AxiomC{$\g \slash \Gamma_{1}, B \vdash \Delta_{1}, \forall x A  \slash \Gamma_{2}, B \vdash \Delta_{2} \slash \h$}
\RightLabel{$\lift$}
\UnaryInfC{$\g \slash \Gamma_{1}, B \vdash \Delta_{1}, \forall x A \slash \Gamma_{2} \vdash \Delta_{2} \slash \h$}
\DisplayProof

&

\AxiomC{$\g \slash \Gamma_{1}, B \vdash \Delta_{1}, \forall x A  \slash \Gamma_{2}, B \vdash \Delta_{2} \slash \h$}
\RightLabel{IH}
\dashedLine
\UnaryInfC{$\g \slash \Gamma_{1}, B \vdash \Delta_{1} \slash \vdash A[\fvx / x] \slash \Gamma_{2}, B \vdash \Delta_{2} \slash \h$}
\dashedLine
\RightLabel{Lem.~\ref{lm:hp-admiss-iw}}
\UnaryInfC{$\g \slash \Gamma_{1}, B \vdash \Delta_{1} \slash B \vdash A[\fvx / x] \slash \Gamma_{2}, B \vdash \Delta_{2} \slash \h$}
\RightLabel{$\lift$}
\UnaryInfC{$\g \slash \Gamma_{1}, B \vdash \Delta_{1} \slash B \vdash A[\fvx / x] \slash \Gamma_{2} \vdash \Delta_{2} \slash \h$}
\RightLabel{$\lift$}
\UnaryInfC{$\g \slash \Gamma_{1}, B \vdash \Delta_{1} \slash \vdash A[\fvx / x] \slash \Gamma_{2} \vdash \Delta_{2} \slash \h$}
\DisplayProof
\end{tabular}
}
\end{center}

If the $\imprtwo$ rule is applied as shown below top, then the desired conclusion may be derived as shown below bottom. In all other cases, we apply IH and then the $\imprtwo$ rule to obtain the desired result.
\begin{center}
\resizebox{\columnwidth}{!}{
\begin{tabular}{c} 
\AxiomC{$\g \slash \Gamma_{1} \vdash \Delta_{1}, \forall x A \slash B \vdash C \slash \Gamma_{2} \vdash \Delta_{2} \slash \h$}
\AxiomC{$\g \slash \Gamma_{1} \vdash \Delta_{1}, \forall x A \slash \Gamma_{2} \vdash \Delta_{2}, B \imp C \slash \h$}
\RightLabel{$\imprtwo$}
\BinaryInfC{$\g \slash \Gamma_{1} \vdash \Delta_{1}, \forall x A, B \imp C \slash \Gamma_{2} \vdash \Delta_{2} \slash \h$}
\DisplayProof
\end{tabular}
}
\end{center}
\begin{center}
\resizebox{\columnwidth}{!}{
\begin{tabular}{c c c}
\AxiomC{$\Pi_{1}$}
\AxiomC{$\Pi_{2}$}
\RightLabel{$\imprtwo$}
\BinaryInfC{$\g \slash \Gamma_{1} \vdash \Delta_{1}, B \imp C \slash \vdash A[\fvx / x] \slash \Gamma_{2} \vdash \Delta_{2} \slash \h$}
\DisplayProof

&

&

\begin{tabular}{c c c}
$\Pi_{1}$

&

$= \Bigg \{$

&

\AxiomC{$\g \slash \Gamma_{1} \vdash \Delta_{1}, \forall x A \slash B \vdash C \slash \Gamma_{2} \vdash \Delta_{2} \slash \h$}
\RightLabel{Lem.~\ref{lm:admiss-lower}}
\dashedLine
\UnaryInfC{$\g \slash \Gamma_{1} \vdash \Delta_{1} \slash B \vdash C, \forall x A \slash \Gamma_{2} \vdash \Delta_{2} \slash \h$}
\RightLabel{IH}
\dashedLine
\UnaryInfC{$\g \slash \Gamma_{1} \vdash \Delta_{1} \slash B \vdash C \slash \vdash A[\fvx / x] \slash \Gamma_{2} \vdash \Delta_{2} \slash \h$}
\DisplayProof
\end{tabular}
\end{tabular}
}
\end{center}
\begin{center}
\resizebox{\columnwidth}{!}{
\begin{tabular}{c c c} 
$\Pi_{2}$

&

$= \Bigg \{$

&

\AxiomC{$\g \slash \Gamma_{1} \vdash \Delta_{1}, \forall x A \slash B \vdash C \slash \Gamma_{2} \vdash \Delta_{2} \slash \h$}
\RightLabel{IH}
\dashedLine
\UnaryInfC{$\g \slash \Gamma_{1} \vdash \Delta_{1} \slash \vdash A[\fvx / x] \slash B \vdash C \slash \Gamma_{2} \vdash \Delta_{2} \slash \h$}

\AxiomC{$\g \slash \Gamma_{1} \vdash \Delta_{1}, \forall x A \slash \Gamma_{2} \vdash \Delta_{2}, B \imp C \slash \h$}
\RightLabel{IH}
\dashedLine
\UnaryInfC{$\g \slash \Gamma_{1} \vdash \Delta_{1} \slash \vdash A[\fvx / x] \slash \Gamma_{2} \vdash \Delta_{2}, B \imp C \slash \h$}

\RightLabel{$\imprtwo$}
\BinaryInfC{$\g \slash \Gamma_{1} \vdash \Delta_{1} \slash \vdash A[\fvx / x], B \imp C \slash \Gamma_{2} \vdash \Delta_{2} \slash \h$}
\DisplayProof
\end{tabular}
}
\end{center}

In the $\allrone$ and $\existsl$ cases, we must ensure that the eigenvariable of the inference is not identical to the parameter $\fvx$ in $A[\fvx / x]$ introduced by IH. However, this can always be ensured by Lem.~\ref{lm:substitution-lemma}. Therefore, we move onto the last nontrivial case, which concerns the $\allrtwo$ rule. The only nontrivial case occurs as shown below top and is resolved as shown below bottom. In all other cases, we apply IH followed by the $\allrtwo$ rule (invoking Lem.~\ref{lm:substitution-lemma} if necessary).
\begin{center}
\resizebox{\columnwidth}{!}{
\begin{tabular}{c} 
\AxiomC{$\g \slash \Gamma_{1} \vdash \Delta_{1}, \forall x A \slash \vdash B[\fvy/y] \slash \Gamma_{2} \vdash \Delta_{2} \slash \h$}
\AxiomC{$\g \slash \Gamma_{1} \vdash \Delta_{1}, \forall x A \slash \Gamma_{2} \vdash \Delta_{2}, \forall y B \slash \h$}
\RightLabel{$\allrtwo$}
\BinaryInfC{$\g \slash \Gamma_{1} \vdash \Delta_{1}, \forall x A, \forall y B \slash \Gamma_{2} \vdash \Delta_{2} \slash \h$}
\DisplayProof
\end{tabular}
}
\end{center}

\begin{center}
\resizebox{\columnwidth}{!}{
\begin{tabular}{c c c}
\AxiomC{$\Pi_{1}$}
\AxiomC{$\Pi_{2}$}
\RightLabel{$\allrtwo$}
\BinaryInfC{$\g \slash \Gamma_{1} \vdash \Delta_{1}, \forall y B \slash \vdash A[\fvx/x] \slash \Gamma_{2} \vdash \Delta_{2} \slash \h$}
\DisplayProof

&

&

\begin{tabular}{c c c}

$\Pi_{1}$

&

$= \Bigg \{$

&

\AxiomC{$\g \slash \Gamma_{1} \vdash \Delta_{1}, \forall x A \slash \vdash B[\fvy/y] \slash \Gamma_{2} \vdash \Delta_{2} \slash \h$}
\RightLabel{Lem.~\ref{lm:substitution-lemma}}
\dashedLine
\UnaryInfC{$\g \slash \Gamma_{1} \vdash \Delta_{1} \slash \vdash B[\fvy/y], \forall x A \slash \Gamma_{2} \vdash \Delta_{2} \slash \h$}
\RightLabel{IH}
\dashedLine
\UnaryInfC{$\g \slash \Gamma_{1} \vdash \Delta_{1} \slash \vdash B[\fvy/y] \slash \vdash A[\fvx / x] \slash \Gamma_{2} \vdash \Delta_{2} \slash \h$}
\DisplayProof
\end{tabular}
\end{tabular}
}
\end{center}
\begin{center}
\resizebox{\columnwidth}{!}{
\begin{tabular}{c c c} 
$\Pi_{2}$

&

$= \Bigg \{$

&

\AxiomC{$\g \slash \Gamma_{1} \vdash \Delta_{1}, \forall x A \slash \vdash B[\fvy/y] \slash \Gamma_{2} \vdash \Delta_{2} \slash \h$}
\RightLabel{IH}
\dashedLine
\UnaryInfC{$\g \slash \Gamma_{1} \vdash \Delta_{1} \slash \vdash A[\fvx / x] \slash \vdash B[\fvy/y] \slash \Gamma_{2} \vdash \Delta_{2} \slash \h$}

\AxiomC{$\g \slash \Gamma_{1} \vdash \Delta_{1}, \forall x A \slash \Gamma_{2} \vdash \Delta_{2}, \forall y B \slash \h$}
\RightLabel{IH}
\dashedLine
\UnaryInfC{$\g \slash \Gamma_{1} \vdash \Delta_{1} \slash \vdash A[\fvx / x] \slash \Gamma_{2} \vdash \Delta_{2}, \forall y B \slash \h$}

\RightLabel{$\allrtwo$}
\BinaryInfC{$\g \slash \Gamma_{1} \vdash \Delta_{1} \slash \vdash A[\fvx / x], \forall y B \slash \Gamma_{2} \vdash \Delta_{2} \slash \h$}
\DisplayProof
\end{tabular}
}
\end{center}
\qed
\end{proof}

\begin{lemma}\label{lm:invert-implies-right-one} The $\improne$ rule is invertible in $\calc$.

\end{lemma}

\begin{proof} We extend the proof of~\cite[Lem.~5.11]{LelKuz18} to include the quantifier cases. The claim is shown by induction on the height of the given derivation. When the last rule of the derivation is $\alll$, $\existsl$, $\existsr$, or $\allrtwo$ in the inductive step, we apply IH to the premise(s) of the inference followed by an application of the corresponding rule. If the last inference of the derivation is an application of the $\allrone$ rule (as shown below left), then the case is resolved as shown below right.
\begin{center}
\resizebox{\columnwidth}{!}{
\begin{tabular}{c c c}
\AxiomC{$\g \slash \Gamma \vdash \Delta, A \imp B \slash \vdash C[\fvx / x]$}
\RightLabel{$\allrone$}
\UnaryInfC{$\g \slash \Gamma \vdash \Delta, A \imp B, \forall x C$}
\DisplayProof

&

&

\AxiomC{$\g \slash \Gamma \vdash \Delta, A \imp B \slash \vdash C[\fvx / x]$}
\RightLabel{Lem.~\ref{lm:admiss-lower}}
\dashedLine
\UnaryInfC{$\g \slash \Gamma \vdash \Delta \slash \vdash C[\fvx / x], A \imp B$}
\RightLabel{IH}
\dashedLine
\UnaryInfC{$\g \slash \Gamma \vdash \Delta \slash \vdash C[\fvx / x] \slash A \vdash B$}

\AxiomC{$\g \slash \Gamma \vdash \Delta, A \imp B \slash \vdash C[\fvx / x]$}
\RightLabel{Lem.~\ref{lm:invert-implies-right-two}}
\dashedLine
\UnaryInfC{$\g \slash \Gamma \vdash \Delta \slash A \vdash B \slash \vdash C[\fvx / x]$}
\RightLabel{$\allrone$}
\UnaryInfC{$\g \slash \Gamma \vdash \Delta \slash A \vdash B, \forall x C$}

\RightLabel{$\allrtwo$}
\BinaryInfC{$\g \slash \Gamma \vdash \Delta, \forall x C \slash A \vdash B$}
\DisplayProof
\end{tabular}
}
\end{center}
\qed
\end{proof}

\begin{lemma}\label{lm:invert-all-right-one} The $\allrone$ rule is invertible in $\calc$.

\end{lemma}

\begin{proof} We prove the result by induction on the height of the given derivation of $\g \slash \Gamma \vdash \Delta, \forall x A$ and show that $\g \slash \Gamma \vdash \Delta \slash \vdash A[\fvx / x]$ is derivable.

\textit{Base case.} If $\g \slash \Gamma \vdash \Delta, \forall x A$ is obtained via $\idone$, $\idtwo$, or $\botl$, then $\g \slash \Gamma \vdash \Delta \slash \vdash A[\fvx / x]$ is an instance of the corresponding rule as well.

\textit{Inductive step.} All cases, with the exception of the $\improne$, $\allrone$, $\existsl$, and $\allrtwo$ rules, are resolved by applying IH to the premise(s) and then applying the relevant rule. Let us consider each of the additional cases in turn.

The $\improne$ case is shown below left and is resolved as shown below right.
\begin{center}
\resizebox{\columnwidth}{!}{
\begin{tabular}{c c c}
\AxiomC{$\g \slash \Gamma \vdash \Delta, \forall x A \slash B \vdash C$}
\RightLabel{$\improne$}
\UnaryInfC{$\g \slash \Gamma \vdash \Delta, \forall x A, B \imp C$}
\DisplayProof

&

&

\AxiomC{$\g \slash \Gamma \vdash \Delta, \forall x A \slash B \vdash C$}
\RightLabel{Lem.~\ref{lm:admiss-lower}}
\dashedLine
\UnaryInfC{$\g \slash \Gamma \vdash \Delta \slash B \vdash C, \forall x A$}
\RightLabel{IH}
\dashedLine
\UnaryInfC{$\g \slash \Gamma \vdash \Delta \slash B \vdash C \slash \vdash A[\fvx / x]$}

\AxiomC{$\g \slash \Gamma \vdash \Delta, \forall x A \slash B \vdash C$}
\RightLabel{Lem.~\ref{lm:invert-all-right-two}}
\dashedLine
\UnaryInfC{$\g \slash \Gamma \vdash \Delta \slash \vdash A[\fvx / x] \slash B \vdash C$}
\RightLabel{$\improne$}
\UnaryInfC{$\g \slash \Gamma \vdash \Delta \slash \vdash A[\fvx / x], B \imp C$}

\RightLabel{$\imprtwo$}
\BinaryInfC{$\g \slash \Gamma \vdash \Delta, B \imp C \slash \vdash A[\fvx / x]$}
\DisplayProof
\end{tabular}
}
\end{center}

In the $\allrone$ case where the relevant formula $\forall x A$ is principal, the premise of the inference is the desired conclusion. If the relevant formula $\forall x A$ is not principal, then the $\allrone$ inference is of the form shown below left and is resolved as shown below right.
\begin{center}
\resizebox{\columnwidth}{!}{
\begin{tabular}{c c c}
\AxiomC{$\g \slash \Gamma \vdash \Delta, \forall x A \slash B[\fvy / y]$}
\RightLabel{$\allrone$}
\UnaryInfC{$\g \slash \Gamma \vdash \Delta, \forall x A, \forall y B$}
\DisplayProof

&

&

\AxiomC{$\g \slash \Gamma \vdash \Delta, \forall x A \slash \vdash B[\fvy / y]$}
\RightLabel{Lem.~\ref{lm:admiss-lower}}
\dashedLine
\UnaryInfC{$\g \slash \Gamma \vdash \Delta \slash B[\fvy / y], \forall x A$}
\RightLabel{IH}
\dashedLine
\UnaryInfC{$\g \slash \Gamma \vdash \Delta \slash \vdash B[\fvy / y] \slash \vdash A[\fvx / x]$}

\AxiomC{$\g \slash \Gamma \vdash \Delta, \forall x A \slash \vdash B[\fvy / y]$}
\RightLabel{Lem.~\ref{lm:invert-all-right-two}}
\dashedLine
\UnaryInfC{$\g \slash \Gamma \vdash \Delta \slash \vdash A[\fvx / x] \slash \vdash B[\fvy / y]$}
\RightLabel{$\allrone$}
\UnaryInfC{$\g \slash \Gamma \vdash \Delta \slash \vdash A[\fvx / x], \forall y B$}

\RightLabel{$\allrtwo$}
\BinaryInfC{$\g \slash \Gamma \vdash \Delta, \forall y B \slash \vdash A[\fvx / x]$}
\DisplayProof

\end{tabular}
}
\end{center}

If the last inference is an instance of the $\existsl$ or $\allrtwo$ rule, then we must ensure that the eigenvariable of the inference is not identical to the parameter $a$ in $A[\fvx / x]$ introduced by IH, but this can always be ensured due to Lem.~\ref{lm:substitution-lemma}.
\qed
\end{proof}

\begin{lemma}\label{lm:admiss-icl}
The $\icl$ rule is admissible in $\calc$.
\end{lemma}

\begin{proof} We extend the proof of \cite[Lem.~5.12]{LelKuz18} and prove the result by induction on the lexicographic ordering of pairs $(|A|,h)$, where $|A|$ is the complexity of the contraction formula $A$ and $h$ is the height of the derivation. We know the result holds for $\lng$, and so, we argue the inductive step for the quantifier rules. 

With the exception of the $\existsl$ case shown below left, all quantifier cases are settled by applying IH followed by an application of the corresponding rule. The only nontrivial case occurs when a contraction is performed on a formula $\exists x A$ with one of the contraction formulae principal in the $\icl$ inference. The situation is resolved as shown below right.
\begin{center}
\begin{tabular}{c @{\hskip 1em} c}
\AxiomC{$\g \slash \Gamma, A[a/x], \exists x A \vdash \Delta \slash \h$}
\RightLabel{$\existsl$}
\UnaryInfC{$\g \slash \Gamma, \exists x A, \exists x A \vdash \Delta \slash \h$}
\RightLabel{$\icl$}
\UnaryInfC{$\g \slash \Gamma, \exists x A \vdash \Delta \slash \h$}
\DisplayProof

&

\AxiomC{$\g \slash \Gamma, A[a/x], \exists x A \vdash \Delta \slash \h$}
\RightLabel{Lem.~\ref{lm:m-invertibility}}
\dashedLine
\UnaryInfC{$\g \slash \Gamma, A[a/x], A[a/x] \vdash \Delta \slash \h$}
\RightLabel{IH}
\dashedLine
\UnaryInfC{$\g \slash \Gamma, A[a/x] \vdash \Delta \slash \h$}
\RightLabel{$\existsl$}
\UnaryInfC{$\g \slash \Gamma, \exists x A \vdash \Delta \slash \h$}
\DisplayProof
\end{tabular}
\end{center}
Notice that IH is applicable since we are contracting on a formula of smaller complexity.
\qed
\end{proof}

\begin{lemma}
\label{lm:admiss-mrg}
The $\mrg$ rule is admissible in $\calc$.
\end{lemma}

\begin{proof} We extend the proof of \cite[Lem.~5.13]{LelKuz18}, which proves that $\mrg$ is admissible in $\lng$, and prove the admissibility of $\mrg$ in $\calc$ by induction on the height of the given derivation. We need only consider the quantifier rules due to \cite[Lem.~5.13]{LelKuz18}. The $\allrone$, $\alll$, $\existsl$, and $\existsr$ cases are all resolved by applying IH to the premise of the rule followed by an application of the rule. If $\mrg$ is applied to the principal components of the $\allrtwo$ rule as follows:
\[
\infer[\mrg]
{\g \sslash \Gamma_{1}, \Gamma_{2} \vdash \Delta_{1}, \Delta_{2}, \forall x A \sslash \h}
{
\infer[\allrtwo]
{\g \sslash \Gamma_{1} \vdash \Delta_{1}, \forall x A \sslash \Gamma_{2} \vdash \Delta_{2} \sslash \h}
{{\g \sslash \Gamma_{1} \vdash \Delta_{1} \sslash \vdash A[\fvx/x] \sslash \Gamma_{2} \vdash \Delta_{2} \sslash \h} & {\g \sslash \Gamma_{1} \vdash \Delta_{1} \sslash \Gamma_{2} \vdash \Delta_{2}, \forall x A \sslash \h}}
}
\]
then the desired conclusion is obtained by applying IH to the top right premise. 
In all other cases, we apply IH to the premises of $\allrtwo$ followed by an application of the rule.
\qed
\end{proof}

\begin{lemma}
\label{lm:admiss-icr}
The $\icr$ rule is admissible in $\calc$.
\end{lemma}

\begin{proof} We extend the proof of \cite[Lem.~5.14]{LelKuz18} to include the quantifier rules and argue the claim by induction on the lexicographic ordering of pairs $(|A|,h)$, where $|A|$ is the complexity of the contraction formula $A$ and $h$ is the height of the derivation. The $\alll$ and $\existsl$ cases are settled by applying IH to the premise of the inference followed by an application of the rule. For the $\existsr$ case, we apply IH to the premise of the rule, followed by the rule. The nontrivial case (occurring when the principal formula is contracted) for the $\allrone$ rule is shown below left, and the desired conclusion is derived as shown below right (where IH is applicable due to the decreased complexity of the contraction formula).
\begin{center}
\begin{tabular}{c @{\hskip 1em} c}
\AxiomC{$\g \slash \Gamma \vdash \Delta, \forall x A \slash \vdash A[\fvx / x]$}
\RightLabel{$\allrone$}
\UnaryInfC{$\g \slash \Gamma \vdash \Delta, \forall x A, \forall x A$}
\RightLabel{$\icr$}
\UnaryInfC{$\g \slash \Gamma \vdash \Delta, \forall x A$}
\DisplayProof

&

\AxiomC{$\g \slash \Gamma \vdash \Delta, \forall x A \slash \vdash A[\fvx / x]$}
\RightLabel{Lem.~\ref{lm:invert-all-right-two}}
\dashedLine
\UnaryInfC{$\g \slash \Gamma \vdash \Delta \slash \vdash A[\fvx / x] \slash \vdash A[\fvx / x]$}
\RightLabel{Lem.~\ref{lm:admiss-mrg}}
\dashedLine
\UnaryInfC{$\g \slash \Gamma \vdash \Delta \slash \vdash A[\fvx / x], A[\fvx / x]$}
\RightLabel{IH}
\dashedLine
\UnaryInfC{$\g \slash \Gamma \vdash \Delta \slash \vdash A[\fvx / x]$}
\RightLabel{$\allrone$}
\UnaryInfC{$\g \slash \Gamma \vdash \Delta, \forall x A$}
\DisplayProof
\end{tabular}
\end{center}
When the contracted formulae are both non-principal in an $\allrone$ inference, we apply IH to the premise followed by an application of the $\allrone$ rule. If the contracted formulae are both non-principal in an $\allrtwo$ inference, then we apply IH to the premises followed by an application of the rule. If one of the contracted formulae is principal in an $\allrtwo$ inference (as shown below top), then the case is settled as shown below bottom.
\begin{center}
\resizebox{\columnwidth}{!}{
\begin{tabular}{c}
\AxiomC{$\g \slash \Gamma_{1} \vdash \Delta_{1}, \forall x A \slash \vdash A[\fvx / x] \slash \Gamma_{2} \vdash \Delta_{2} \slash \h$}
\AxiomC{$\g \slash \Gamma_{1} \vdash \Delta_{1}, \forall x A \slash \Gamma_{2} \vdash \Delta_{2}, \forall x A \slash \h$}
\RightLabel{$\allrtwo$}
\BinaryInfC{$\g \slash \Gamma_{1} \vdash \Delta_{1}, \forall x A, \forall x A \slash \Gamma_{2} \vdash \Delta_{2} \slash \h$}
\DisplayProof
\end{tabular}
}
\end{center}
\begin{center}
\resizebox{\columnwidth}{!}{
\begin{tabular}{c} 
\AxiomC{$\g \slash \Gamma_{1} \vdash \Delta_{1}, \forall x A \slash \vdash A[\fvx / x] \slash \Gamma_{2} \vdash \Delta_{2} \slash \h$}
\RightLabel{Lem.~\ref{lm:invert-all-right-two}}
\dashedLine
\UnaryInfC{$\g \slash \Gamma_{1} \vdash \Delta_{1} \slash \vdash A[\fvx / x] \slash \vdash A[\fvx / x] \slash \Gamma_{2} \vdash \Delta_{2} \slash \h$}
\RightLabel{Lem.~\ref{lm:admiss-mrg}}
\dashedLine
\UnaryInfC{$\g \slash \Gamma_{1} \vdash \Delta_{1} \slash \vdash A[\fvx / x], A[\fvx / x] \slash \Gamma_{2} \vdash \Delta_{2} \slash \h$}
\RightLabel{IH}
\dashedLine
\UnaryInfC{$\g \slash \Gamma_{1} \vdash \Delta_{1} \slash \vdash A[\fvx / x] \slash \Gamma_{2} \vdash \Delta_{2} \slash \h$}

\AxiomC{$\g \slash \Gamma_{1} \vdash \Delta_{1}, \forall x A \slash \Gamma_{2} \vdash \Delta_{2}, \forall x A \slash \h$}
\RightLabel{Lem.~\ref{lm:admiss-lower}}
\dashedLine
\UnaryInfC{$\g \slash \Gamma_{1} \vdash \Delta_{1} \slash \Gamma_{2} \vdash \Delta_{2}, \forall x A, \forall x A \slash \h$}
\RightLabel{IH}
\dashedLine
\UnaryInfC{$\g \slash \Gamma_{1} \vdash \Delta_{1} \slash \Gamma_{2} \vdash \Delta_{2}, \forall x A \slash \h$}

\RightLabel{$\allrtwo$}
\BinaryInfC{$\g \slash \Gamma_{1} \vdash \Delta_{1}, \forall x A \slash \Gamma_{2} \vdash \Delta_{2} \slash \h$}
\DisplayProof
\end{tabular}
}
\end{center}
Note that we may apply IH in the left branch of the derivation since the complexity of the contraction formula is less than $\forall x A$, and we may apply IH in the right branch since the height of the derivation is less than the original.
\qed
\end{proof}

Before moving on to the cut-elimination theorem, we present the definition of the \emph{splice} operation~\cite{LelKuz18,Mas92}. The operation is used to formulate the $\cut$ rule.

\begin{definition}[Splice~\cite{LelKuz18}] The \emph{splice} $\g \oplus \h$ of two linear nested sequents $\g$ and $\h$ is defined as follows:

\begin{center}
$(\Gamma_{1} \vdash \Delta_{1}) \oplus (\Gamma_{2} \vdash \Delta_{2}) := \Gamma_{1}, \Gamma_{2} \vdash \Delta_{1}, \Delta_{2}$\\
$(\Gamma_{1} \vdash \Delta_{1}) \oplus (\Gamma_{2} \vdash \Delta_{2} \sslash \f) := \Gamma_{1}, \Gamma_{2} \vdash \Delta_{1}, \Delta_{2} \sslash \f$\\
$(\Gamma_{1} \vdash \Delta_{1} \sslash \f) \oplus (\Gamma_{2} \vdash \Delta_{2}) := \Gamma_{1}, \Gamma_{2} \vdash \Delta_{1}, \Delta_{2} \sslash \f$\\
$(\Gamma_{1} \vdash \Delta_{1} \sslash \f) \oplus (\Gamma_{2} \vdash \Delta_{2} \sslash \k) := \Gamma_{1}, \Gamma_{2} \vdash \Delta_{1}, \Delta_{2} \sslash (\f \oplus \k)$\\
\end{center}

\end{definition}

\begin{theorem}[Cut-Elimination]\label{thm:cut-elimination} The rule 
\begin{center}
\AxiomC{$\g \sslash \Gamma \vdash \Delta, A \sslash \h$}
\AxiomC{$\f \sslash A^{k_{1}}, \Gamma_{1} \vdash \Delta_{1} \sslash \cdots \sslash A^{k_{n}}, \Gamma_{n} \vdash \Delta_{n}$}
\RightLabel{$\cut$}
\BinaryInfC{$(\g \oplus \f) \sslash \Gamma, \Gamma_{1} \vdash \Delta, \Delta_{1} \slash \big(\h \oplus (\Gamma_{2} \vdash \Delta_{2} \slash \cdots \slash \Gamma_{n} \vdash \Delta_{n})\big)$}
\DisplayProof
\end{center}
where $\parallel \g \parallel \ = \ \parallel \f \parallel$, $\parallel \h \parallel \ = \ n-1$, and $\sum_{i=1}^{n} k_{i} \geq 1$, is eliminable in $\calc$.
\end{theorem}

\begin{proof} We extend the proof of \cite[Thm.~5.16]{LelKuz18} and prove the result by induction on the lexicographic ordering of pairs $(|A|,h_{1},h_{2})$, where $|A|$ is the complexity of the cut formula $A$, $h_{1}$ is the height of the derivation of the \emph{right premise} of the $\cut$ rule, and $h_{2}$ is the height of the derivation of the \emph{left premise} of the $\cut$ rule. Moreover, we assume w.l.o.g. that $\cut$ is used once as the last inference of the derivation (given a derivation with multiple applications of $\cut$, we may repeatedly apply the elimination algorithm described here to the topmost occurrence of $\cut$, ultimately resulting in a cut-free derivation). By \cite[Thm.~5.16]{LelKuz18}, we know that $\cut$ is eliminable from any derivation in $\lng$, and therefore, we need only consider cases which incorporate quantifier rules.

If $h_{1} = 0$, then the right premise of $\cut$ is an instance of $\idone$, $\idtwo$, or $\botl$. If none of the cut formulae $A$ are principal in the right premise, then the conclusion of $\cut$ is an instance of $\idone$, $\idtwo$, or $\botl$. If, however, one of the cut formulae $A$ is principal in the right premise and is an atomic formula $p(\vv{a})$, then the top right premise of $\cut$ is of the form
$$
\f \sslash p(\vv{a})^{k_{1}}, \Gamma_{1} \vdash \Delta_{1} \sslash \cdots \slash p(\vv{a})^{k_{i}}, \Gamma_{i} \vdash p(\vv{a}), \Delta_{i}' \slash \cdots \sslash p(\vv{a})^{k_{n}}, \Gamma_{n} \vdash \Delta_{n}
$$
where $\Delta_{i} = p(\vv{a}), \Delta_{i}'$. Observe that since $\Delta_{i}$ occurs in the conclusion of $\cut$, so does $p(\vv{a})$. To construct a cut-free derivation of the conclusion of $\cut$, we apply $\lwr$ to the left premise $\g \sslash \Gamma \vdash \Delta, p(\vv{a}) \sslash \h$ until $p(\vv{a})$ is in the $i^{th}$ component, and then apply hp-admissibility of $\iw$ (Lem.~\ref{lm:hp-admiss-iw}) to add in the missing formulae. Last, if the cut formula $A$ is principal in the right premise and is equal to $\bot$, then the left premise of $\cut$ is of the form $\g \sslash \Gamma \vdash \Delta, \bot \sslash \h$. We obtain a cut-free derivation of the conclusion of $\cut$ by first applying hp-admissibility of $(\bot_{r})$ (Lem.~\ref{lem:botr-admiss}), followed by hp-admissibility of $\iw$ (Lem.~\ref{lm:hp-admiss-iw}) to add in the missing formulae.

Suppose that $h_{1} > 0$. If none of the cut formulae $A$ are principal in the inference $(r)$ of the right premise of $\cut$, then for all cases (with the exception of the $\allrone$, $\improne$, $\existsl$, $\allrtwo$, and $\imprtwo$ cases) we apply IH with the left premise of $\cut$ and the premise(s) of $(r)$, followed by an application of $(r)$. Let us now consider the $\allrone$, $\existsl$, $\allrtwo$,  $\improne$, and $\imprtwo$ cases when none of the cut formulae $A$ are principal. First, assume that $\allrone$ is the rule used to derive the right premise of $\cut$:
\begin{center}
\AxiomC{$\g \sslash \Gamma \vdash \Delta, A \sslash \h$}

\AxiomC{$\f \sslash A^{k_{1}}, \Gamma_{1} \vdash \Delta_{1} \sslash \cdots \sslash A^{k_{n}}, \Gamma_{n} \vdash \Delta_{n} \slash \vdash B[\fvx / x]$}
\RightLabel{$\allrone$}
\UnaryInfC{$\f \sslash A^{k_{1}}, \Gamma_{1} \vdash \Delta_{1} \sslash \cdots \sslash A^{k_{n}}, \Gamma_{n} \vdash \Delta_{n}, \forall x B$}
\RightLabel{$\cut$}
\BinaryInfC{$(\g \oplus \f) \sslash \Gamma, \Gamma_{1} \vdash \Delta, \Delta_{1} \slash \big(\h \oplus (\Gamma_{2} \vdash \Delta_{2} \slash \cdots \slash \Gamma_{n} \vdash \Delta_{n}, \forall x B)\big)$}
\DisplayProof
\end{center}
We invoke hp-admissibility of $\vsub$ (Lem.~\ref{lm:substitution-lemma}) to substitute the eigenvariable $\fvx$ of $\allrone$ with a fresh variable $\fvy$ that does not occur in either premise of $\cut$. We then apply admissibility of $\ew$ (Lem.~\ref{lm:admiss-ew}) to the left premise of $\cut$, apply IH to the resulting derivations, and last apply the $\allrone$ rule, as shown below:
\begin{center}
\resizebox{\columnwidth}{!}{
\AxiomC{$\g \sslash \Gamma \vdash \Delta, A \sslash \h$}
\RightLabel{Lem.~\ref{lm:admiss-ew}}
\dashedLine
\UnaryInfC{$\g \sslash \Gamma \vdash \Delta, A \sslash \h \slash \vdash $}

\AxiomC{$\f \sslash A^{k_{1}}, \Gamma_{1} \vdash \Delta_{1} \sslash \cdots \sslash A^{k_{n}}, \Gamma_{n} \vdash \Delta_{n} \slash \vdash B[\fvx / x]$}
\RightLabel{Lem.~\ref{lm:substitution-lemma}}
\dashedLine
\UnaryInfC{$\f \sslash A^{k_{1}}, \Gamma_{1} \vdash \Delta_{1} \sslash \cdots \sslash A^{k_{n}}, \Gamma_{n} \vdash \Delta_{n} \slash \vdash B[\fvy / x]$}

\RightLabel{IH}
\dashedLine
\BinaryInfC{$(\g \oplus \f) \sslash \Gamma, \Gamma_{1} \vdash \Delta, \Delta_{1} \slash \big(\h \oplus (\Gamma_{2} \vdash \Delta_{2} \slash \cdots \slash \Gamma_{n} \vdash \Delta_{n})\big) \slash \vdash B[\fvy / x] $}

\RightLabel{$\allrone$}
\UnaryInfC{$(\g \oplus \f) \sslash \Gamma, \Gamma_{1} \vdash \Delta, \Delta_{1} \slash \big(\h \oplus (\Gamma_{2} \vdash \Delta_{2} \slash \cdots \slash \Gamma_{n} \vdash \Delta_{n}, \forall x B)\big)$}

\DisplayProof
}
\end{center}

In the $\existsl$ case below
\begin{center}
\resizebox{\columnwidth}{!}{
\begin{tabular}{c} 
\AxiomC{$\g \sslash \Gamma \vdash \Delta, A \sslash \h$}

\AxiomC{$\f \sslash A^{k_{1}}, \Gamma_{1} \vdash \Delta_{1} \sslash \cdots \slash A^{k_{i}}, B[\fvx / x], \Gamma_{i} \vdash \Delta_{i} \slash \cdots \slash A^{k_{n}}, \Gamma_{n} \vdash \Delta_{n}$}
\RightLabel{$\existsl$}
\UnaryInfC{$\f \sslash A^{k_{1}}, \Gamma_{1} \vdash \Delta_{1} \sslash \cdots \slash A^{k_{i}}, \exists x B, \Gamma_{i} \vdash \Delta_{i} \slash \cdots \slash A^{k_{n}}, \Gamma_{n} \vdash \Delta_{n}$}
\RightLabel{$\cut$}
\BinaryInfC{$(\g \oplus \f) \sslash \Gamma, \Gamma_{1} \vdash \Delta, \Delta_{1} \slash \big(\h \oplus (\Gamma_{2} \vdash \Delta_{2} \sslash \cdots \slash \exists x B, \Gamma_{i} \vdash \Delta_{i} \slash \cdots \slash \Gamma_{n} \vdash \Delta_{n})\big)$}
\DisplayProof
\end{tabular}
}
\end{center}
we also make use of the hp-admissibility of $\vsub$ to ensure that the $\existsl$ rule can be applied after invoking the inductive hypothesis:
\begin{center}
\resizebox{\columnwidth}{!}{
\begin{tabular}{c} 
\AxiomC{$\g \sslash \Gamma \vdash \Delta, A \sslash \h$}

\AxiomC{$\f \sslash A^{k_{1}}, \Gamma_{1} \vdash \Delta_{1} \sslash \cdots \slash A^{k_{i}}, B[\fvx / x], \Gamma_{i} \vdash \Delta_{i} \slash \cdots \slash A^{k_{n}}, \Gamma_{n} \vdash \Delta_{n}$}
\RightLabel{Lem.~\ref{lm:substitution-lemma}}
\dashedLine
\UnaryInfC{$\f \sslash A^{k_{1}}, \Gamma_{1} \vdash \Delta_{1} \sslash \cdots \slash A^{k_{i}}, B[\fvy / x], \Gamma_{i} \vdash \Delta_{i} \slash \cdots \slash A^{k_{n}}, \Gamma_{n} \vdash \Delta_{n}$}

\RightLabel{IH}
\dashedLine
\BinaryInfC{$(\g \oplus \f) \sslash \Gamma, \Gamma_{1} \vdash \Delta, \Delta_{1} \slash \big(\h \oplus (\Gamma_{2} \vdash \Delta_{2} \sslash \cdots \slash A^{k_{i}}, B[\fvy / x], \Gamma_{i} \vdash \Delta_{i} \slash \cdots \slash \Gamma_{n} \vdash \Delta_{n})\big)$}

\RightLabel{$\existsl$}
\UnaryInfC{$(\g \oplus \f) \sslash \Gamma, \Gamma_{1} \vdash \Delta, \Delta_{1} \slash \big(\h \oplus (\Gamma_{2} \vdash \Delta_{2} \sslash \cdots \slash \exists x B, \Gamma_{i} \vdash \Delta_{i} \slash \cdots \slash \Gamma_{n} \vdash \Delta_{n})\big)$}
\DisplayProof
\end{tabular}
}
\end{center}

Let us consider the $\allrtwo$ case
\begin{center}
\resizebox{\columnwidth}{!}{
\begin{tabular}{c} 
\AxiomC{(1)}

\AxiomC{(2)}
\AxiomC{(3)}
\RightLabel{$\allrtwo$}
\BinaryInfC{$\f \sslash A^{k_{1}}, \Gamma_{1} \vdash \Delta_{1} \slash \cdots \sslash A^{k_{i}}, \Gamma_{i} \vdash \Delta_{i}, \forall x B \slash A^{k_{i+1}}, \Gamma_{i+1} \vdash \Delta_{i+1} \slash \cdots \slash A^{k_{n}}, \Gamma_{n} \vdash \Delta_{n}$}
\RightLabel{$\cut$}
\BinaryInfC{$(\g \oplus \f) \sslash \Gamma, \Gamma_{1} \vdash \Delta, \Delta_{1} \slash \big(\h \oplus (\Gamma_{2} \vdash \Delta_{2} \slash \cdots \sslash \Gamma_{i} \vdash \Delta_{i}, \forall x B \slash \Gamma_{i+1} \vdash \Delta_{i+1} \slash \cdots \slash \Gamma_{n} \vdash \Delta_{n})\big)$}
\DisplayProof
\end{tabular}
}
\end{center}
\begin{scriptsize}
$$
\begin{aligned}
(1) & \ \ \g \sslash \Gamma \vdash \Delta, A \sslash \h_{1} \slash \Gamma_{i}' \vdash \Delta_{i}' \slash \Gamma_{i+1}' \vdash \Delta_{i+1}' \slash \h_{2}\\
(2) & \ \ \f \slash A^{k_{1}}, \Gamma_{1} \vdash \Delta_{1} \slash \cdots \sslash A^{k_{i}}, \Gamma_{i} \vdash \Delta_{i} \slash \vdash B[\fvx / x] \slash A^{k_{i+1}}, \Gamma_{i+1} \vdash \Delta_{i+1} \slash \cdots \slash A^{k_{n}}, \Gamma_{n} \vdash \Delta_{n}\\
(3) & \ \ \f \slash A^{k_{1}}, \Gamma_{1} \vdash \Delta_{1} \slash \cdots \sslash A^{k_{i}}, \Gamma_{i} \vdash \Delta_{i} \slash A^{k_{i+1}}, \Gamma_{i+1} \vdash \Delta_{i+1}, \forall x B \slash \cdots \slash A^{k_{n}}, \Gamma_{n} \vdash \Delta_{n}
\end{aligned}
$$
\end{scriptsize}

\noindent
where $\h = \h_{1} \slash \Gamma_{i}' \vdash \Delta_{i}' \slash \Gamma_{i+1}' \vdash \Delta_{i+1}' \slash \h_{2}$. To resolve the case we invoke admissibility of $\ew$ (Lem.~\ref{lm:admiss-ew}) on (1) to obtain a derivation of
$$
(1)' \ \ \g \sslash \Gamma \vdash \Delta, A \sslash \h_{1} \slash \Gamma_{i}' \vdash \Delta_{i}' \slash \vdash \slash \Gamma_{i+1}' \vdash \Delta_{i+1}' \slash \h_{2}
$$
Moreover, to ensure that the eigenvariable $\fvx$ in (2) does not occur in (1), we apply hp-admissibility of $\vsub$ (Lem.~\ref{lm:substitution-lemma}) to obtain $(2)'$ where $\fvx$ has been replaced by a fresh parameter $\fvy$. Applying IH between $(1)'$ and $(2)'$, and (1) and (3), followed by an application of $\allrtwo$, gives the desired result. Last, note that the $\improne$ and $\imprtwo$ cases are resolved as explained in the proof of \cite[Thm.~5.16]{LelKuz18}.

We assume now that one of the cut formulae $A$ is principal in the inference yielding the right premise of $\cut$. The cases where $A$ is an atomic formula $p(\vv{a})$ or is identical to $\bot$ are resolved as explained above (when $h_{1} = 0$). For the case when $A$ is principal in an application of $\lift$, we simply apply IH between the left premise of $\cut$ and the premise of the $\lift$ rule. Also, if $A$ is of the form $B \land C$, $B \lor C$, or $B \imp C$, then all such cases can be resolved as explained in the proof of \cite[Thm.~5.16]{LelKuz18}. Thus, we only consider the cases where $A$ is of the form $\exists x B$ and $\forall x B$. We first consider the former case, which can be split into two subcases: either the cut formula $\exists x B$ is principal in the left premise of $\cut$, or it is not. Let us consider the former case first, where the cut formula $\exists x B$ is principal in the left premise:
\begin{center}
\resizebox{\columnwidth}{!}{
\begin{tabular}{c} 
\AxiomC{$\g \sslash \Gamma \vdash \Delta, B[\fvy/x], \exists x B \sslash \h$}
\RightLabel{$\existsr$}
\UnaryInfC{$\g \sslash \Gamma \vdash \Delta, \exists x B \sslash \h$}

\AxiomC{$\f \sslash \exists x B^{k_{1}}, \Gamma_{1} \vdash \Delta_{1} \sslash \cdots \slash \exists x B^{k_{i}}, B[\fvx / x], \Gamma_{i} \vdash \Delta_{i} \slash \cdots \slash \exists x B^{k_{n}}, \Gamma_{n} \vdash \Delta_{n}$}
\RightLabel{$\existsl$}
\UnaryInfC{$\f \sslash \exists x B^{k_{1}}, \Gamma_{1} \vdash \Delta_{1} \sslash \cdots \slash \exists x B^{k_{i+1}}, \Gamma_{i} \vdash \Delta_{i} \slash \cdots \slash \exists x B^{k_{n}}, \Gamma_{n} \vdash \Delta_{n}$}
\RightLabel{$\cut$}
\BinaryInfC{$(\g \oplus \f) \sslash \Gamma, \Gamma_{1} \vdash \Delta, \Delta_{1} \slash \big(\h \oplus (\Gamma_{2} \vdash \Delta_{2} \sslash \cdots \slash  \Gamma_{i} \vdash \Delta_{i} \slash \cdots \slash \Gamma_{n} \vdash \Delta_{n})\big)$}
\DisplayProof
\end{tabular}
}
\end{center}
By IH, the premise of $\existsr$ and the right premise of $\cut$ yield a cut-free derivation of:
$$
(\g \oplus \f) \sslash \Gamma, \Gamma_{1} \vdash B[\fvy / x], \Delta, \Delta_{1} \slash \big(\h \oplus (\Gamma_{2} \vdash \Delta_{2} \sslash \cdots \slash  \Gamma_{i} \vdash \Delta_{i} \slash \cdots \slash \Gamma_{n} \vdash \Delta_{n})\big)
$$
By hp-admissibility of $\vsub$ (Lem.~\ref{lm:substitution-lemma}), we have a derivation of the premise of $\existsl$ where the eigenvariable $\fvx$ has been replaced by $\fvy$. Invoking IH between the derivation of this sequent together with the left premise of $\cut$ gives a cut-free derivation of:
$$
(\g \oplus \f) \sslash \Gamma, \Gamma_{1} \vdash \Delta, \Delta_{1} \slash \big(\h \oplus (\Gamma_{2} \vdash \Delta_{2} \sslash \cdots \slash  B[\fvy / x], \Gamma_{i} \vdash \Delta_{i} \slash \cdots \slash \Gamma_{n} \vdash \Delta_{n})\big)
$$
Since $|B[\fvy / x]| < |\exists x B|$, we can apply IH to the derivations of the above two sequents. After applying admissibility of $\icl$ and $\icr$ (Lem.~\ref{lm:admiss-icl} and~\ref{lm:admiss-icr}), we obtain the desired conclusion.

Let us now suppose that $\exists x B$ is not principal in the left premise of $\cut$ and that the left premise is obtained by an instance of the rule $(r)$, that is, our derivation ends with:
\begin{center}
\resizebox{\columnwidth}{!}{
\begin{tabular}{c} 
\AxiomC{$\g \sslash \Gamma \vdash \Delta, \exists x B \sslash \h$}

\AxiomC{$\f \sslash \exists x B^{k_{1}}, \Gamma_{1} \vdash \Delta_{1} \sslash \cdots \slash \exists x B^{k_{i}}, B[\fvx / x], \Gamma_{i} \vdash \Delta_{i} \slash \cdots \slash \exists x B^{k_{n}}, \Gamma_{n} \vdash \Delta_{n}$}
\RightLabel{$\existsl$}
\UnaryInfC{$\f \sslash \exists x B^{k_{1}}, \Gamma_{1} \vdash \Delta_{1} \sslash \cdots \slash \exists x B^{k_{i+1}}, \Gamma_{i} \vdash \Delta_{i} \slash \cdots \slash \exists x B^{k_{n}}, \Gamma_{n} \vdash \Delta_{n}$}
\RightLabel{$\cut$}
\BinaryInfC{$(\g \oplus \f) \sslash \Gamma, \Gamma_{1} \vdash \Delta, \Delta_{1} \slash \big(\h \oplus (\Gamma_{2} \vdash \Delta_{2} \sslash \cdots \slash  \Gamma_{i} \vdash \Delta_{i} \slash \cdots \slash \Gamma_{n} \vdash \Delta_{n})\big)$}
\DisplayProof
\end{tabular}
}
\end{center}
If $(r)$ is a rule other than $\imprtwo$ or $\allrtwo$, then we obtain the desired conclusion by applying a $\cut$ between the premise(s) of $(r)$ and the right premise of $\cut$, followed by an application of $(r)$. If $(r)$ is either the $\imprtwo$ or $\allrtwo$ rule, then we invoke IH between the right premise of $(r)$ and the right premise of $\cut$ to obtain $\g_{1}$, apply admissibility of $\ew$ (Lem.~\ref{lm:admiss-ew}) to the right premise of $\cut$ to weaken in the empty component $\vdash$ in the appropriate place so that IH can be applied to the resulting sequent and the left premise of $(r)$ giving $\g_{2}$, and last, we obtain the desired concusion by applying $(r)$ to $\g_{1}$ and $\g_{2}$. (NB. If $(r)$ is a rule with an eigenvariable condition, then it may be necessary to first apply hp-admissibility of $\vsub$ (Lem.~\ref{lm:substitution-lemma}) to the relevant premise of $(r)$ to ensure the satisfaction of the condition after the $\cut$ has been applied.)

We obtain the desired result by applying a $\cut$ between the premise(s) of $(r)$ and the right premise of $\cut$, followed by an application of $(r)$. Note that if $(r)$ is a rule with an eigenvariable condition, then it may be necessary to first apply hp-admissibility of $\vsub$ (Lem.~\ref{lm:substitution-lemma}) to the premise of $(r)$ to ensure the satisfaction of the condition after the $\cut$ has been applied. 

Last, let us consider the case where $A$ is of the form $\forall x B$:
\begin{center}
\resizebox{\columnwidth}{!}{
\begin{tabular}{c} 
\AxiomC{$\g \sslash \Gamma \vdash \Delta, \forall x B \sslash \h$}

\AxiomC{$\f \sslash \forall x B^{k_{1}}, \Gamma_{1} \vdash \Delta_{1} \sslash \cdots \slash \forall x B^{k_{i}}, B[a / x], \Gamma_{i} \vdash \Delta_{i} \slash \cdots \slash \forall x B^{k_{n}}, \Gamma_{n} \vdash \Delta_{n}$}
\RightLabel{$\alll$}
\UnaryInfC{$\f \sslash \forall x B^{k_{1}}, \Gamma_{1} \vdash \Delta_{1} \sslash \cdots \slash \forall x B^{k_{i}}, \Gamma_{i} \vdash \Delta_{i} \slash \cdots \slash \forall x B^{k_{n}}, \Gamma_{n} \vdash \Delta_{n}$}
\RightLabel{$\cut$}
\BinaryInfC{$(\g \oplus \f) \sslash \Gamma, \Gamma_{1} \vdash \Delta, \Delta_{1} \slash \big(\h \oplus (\Gamma_{2} \vdash \Delta_{2} \sslash \cdots \slash  \Gamma_{i} \vdash \Delta_{i} \slash \cdots \slash \Gamma_{n} \vdash \Delta_{n})\big)$}
\DisplayProof
\end{tabular}
}
\end{center}
Applying IH between the left premise of $\cut$ and the premise of the $\alll$ rule, we obtain
$$
(\g \oplus \f) \sslash \Gamma, \Gamma_{1} \vdash \Delta, \Delta_{1} \slash \big(\h \oplus (\Gamma_{2} \vdash \Delta_{2} \sslash \cdots \slash  B[a / x], \Gamma_{i} \vdash \Delta_{i} \slash \cdots \slash \Gamma_{n} \vdash \Delta_{n})\big)
$$
Depending on if $\h$ is empty or not, we invoke the invertibility of $\allrone$ or $\allrtwo$ (Lem.~\ref{lm:invert-all-right-one} and~\ref{lm:invert-all-right-two}), admissibility of $\mrg$ (Lem.~\ref{lm:admiss-mrg}), and hp-admissibility of $\vsub$ (Lem.~\ref{lm:substitution-lemma}) to obtain a derivation of the sequent $\g \slash \Gamma \vdash \Delta, B[a/x] \slash \h$. Since $|B[\fvx / x]| < |\forall x B|$ we can apply IH between this sequent and the one above to obtain a cut-free derivation of:
$$
(\g \oplus \g \oplus \f) \sslash \Gamma, \Gamma, \Gamma_{1} \vdash \Delta, \Delta, \Delta_{1} \slash \big(\h \oplus \h \oplus (\Gamma_{2} \vdash \Delta_{2} \sslash \cdots \slash \Gamma_{i} \vdash \Delta_{i} \slash \cdots \slash \Gamma_{n} \vdash \Delta_{n})\big)
$$
Admissibility of $\icl$ and $\icr$ (Lem.~\ref{lm:admiss-icl} and~\ref{lm:admiss-icr}) give the desired conclusion.
\qed
\end{proof}



\section{Conclusion}\label{sec:conclusion}

This paper presented the cut-free calculus $\calc$ for intuitionistic fuzzy logic within the relatively new paradigm of \emph{linear nested sequents}. The calculus possesses highly fundamental proof-theoretic properties such as (m-)invertibility of all logical rules, admissibility of structural rules, and syntactic cut-elimination.

In future work the author aims to investigate corollaries of the cut-elimination theorem, such as a 
\emph{midsequent theorem}~\cite{BaaZac00}. In our context, such a theorem states that every derivable sequent containing only prenex formulae is derivable with a proof containing quantifier-free sequents, called \emph{midsequents}, which have only propositional inferences (and potentially $\lift$) above them in the derivation, and only quantifier inferences (and potentially $\lift$) below them. 
Moreover, the present formalism could offer insight regarding which fragments interpolate (or if all of $\ifl$ interpolates) by applying the so-called \emph{proof-theoretic method} of interpolation~\cite{LelKuz18,LyoTiuGorClo20}. 
Additionally, it could be fruitful to adapt linear nested sequents to other first-order G\"odel logics and to investigate decidable fragments~\cite{BaaCiaPre09} by providing proof-search algorithms with implementations (e.g.~\cite{Lel16} provides an implementation of proof-search in \textsc{Prolog} for a class of modal logics within the linear nested sequent framework).  

Last, \cite{Fit14} introduced both a nested calculus for first-order intuitionistic logic with \emph{constant domains}, and a nested calculus for first-order intuitionistic logic with \emph{non-constant domains}. The fundamental difference between the two calculi involves the imposition of a side condition on the left $\forall$ and right $\exists$ rules. The author aims to investigate whether such a condition can be imposed on quantifier rules in $\calc$ in order to readily convert the calculus into a sound and cut-free complete calculus for first-order G\"odel logic with \emph{non-constant domains}. This would be a further strength of $\calc$ since switching between the calculi for the constant domain and non-constant domain versions of first-order G\"odel logic would result by simply imposing a side condition on a subset of the quantifier rules. 




\ \\
\noindent
\acknowledgments{The author would like to thank his supervisor A. Ciabattoni for her continued support, B. Lellmann for his thought-provoking discussions on linear nested sequents, and K. van Berkel for his helpful comments.}

\bibliographystyle{abbrv}
\bibliography{mybib}


\appendix

\section{Proofs}\label{appendix}

\begin{customthm}{\ref{thm:soundness-calc}}[Soundness of $\calc$]
For any linear nested sequent $\g$, if $\g$ is provable in $\calc$, then $\Vdash \uc \iota(\g)$.
\end{customthm}

\begin{proof} We prove the result by induction on the height of the derivation of $\g$.

\textit{Base case.} We argue by contradiction that $(r) \in \{\idone, \idtwo, \botl\}$ always produces a valid linear nested sequent under the universal closure of $\iota$. Each rule is of the form shown below left with the sequent $\g$ of the form shown below right:
\begin{center}
\begin{tabular}{c @{\hskip 3em} c}
\AxiomC{}
\RightLabel{$(r)$}
\UnaryInfC{$\g$}
\DisplayProof

&

$\g = \Gamma_{1} \vdash \Delta_{1} \sslash \cdots \sslash \Gamma_{n} \vdash \Delta_{n} \sslash \cdots \sslash \Gamma_{m} \vdash \Delta_{m}$
\end{tabular}
\end{center}
We therefore assume that $\g$ is invalid. This implies that there exists a model $M = (W,R,D,V)$ with world $v$ such that $Rvw_{0}$, $\vv{a} \in D_{w_{0}}$, and $M, w_{0} \not\Vdash \iota(\g)(\vv{a})$. It follows that there is a sequence of worlds $w_{1}, \cdots, w_{m} \in W$ such that $Rw_{j}w_{j+1}$ (for $0 \leq j \leq m-1$), $M, w_{i} \Vdash \bigwedge \Gamma_{i}$, and $M, w_{i} \not\Vdash \bigvee \Delta_{i}$, for each $1 \leq i \leq m$. We assume all parameters in $\bigwedge \Gamma_{i}$ and $\bigvee \Delta_{i}$ (for $1 \leq i \leq m$) are interpreted as elements of the associated domain.

{\textbf{$\idone$-rule:}} Let $\g$ be in the form above and assume that $\Gamma_{n} = \Gamma_{n}', p(\vv{\fvy})$ and $\Delta_{n} = p(\vv{\fvy}), \Delta_{n}'$. If this happens to be the case, then $M, w_{n} \Vdash \bigwedge \Gamma_{n}' \wedge p(\vv{\fvy})$ and $M, w_{n} \not\Vdash \bigvee \Delta_{n}' \vee p(\vv{\fvy})$. The former implies that $M, w_{n} \Vdash p(\vv{\fvy})$ and the latter implies that $M, w_{n} \not\Vdash p(\vv{\fvy})$, which is a contradiction.

{\textbf{$\idtwo$-rule:}} Similar to the case above, but uses Lem.~\ref{lm:persistency}.

{\textbf{$(\bot_{l})$-rule:}} Let $\g$ be as above and assume that $\Gamma_{n} = \Gamma_{n}', \bot$. If this is the case, then it follows that $M, w_{n} \Vdash \bigwedge \Gamma_{i} \wedge \bot$, from which the contradiction $M, w_{n} \Vdash \bot$ follows.

\textit{Inductive step.} Each inference rule considered is of one of the following two forms: 
\begin{center}
\begin{tabular}{c @{\hskip 5em} c}
\AxiomC{$\g'$}
\RightLabel{$(r_{1})$}
\UnaryInfC{$\g$}
\DisplayProof

&

\AxiomC{$\g_{1}$}
\AxiomC{$\g_{2}$}
\RightLabel{$(r_{2})$}
\BinaryInfC{$\g$}
\DisplayProof
\end{tabular}
\end{center}

\begin{center}
\begin{tabular}{c}
where $\g = \Gamma_{1} \vdash \Delta_{1} \sslash \cdots \sslash \Gamma_{n} \vdash \Delta_{n} \sslash \Gamma_{n+1} \vdash \Delta_{n+1} \sslash \cdots \sslash \Gamma_{m} \vdash \Delta_{m}$.
    
\end{tabular}
\end{center}
Assuming that $\g$ is invalid implies the existence of a model $M = (W,R,D,V)$ with world $v$ such that $Rvw_{0}$, $\vv{a} \in D_{w_{0}}$, and $M, w_{0} \not\Vdash \iota(\g) \sigma$. Hence, there is a sequence of worlds $w_{1}, \cdots, w_{m} \in W$ such that $Rw_{j}w_{j+1}$ (for $0 \leq j \leq m-1$), $M, w_{i} \Vdash \bigwedge \Gamma_{i}$, and $M, w_{i} \not\Vdash \bigvee \Delta_{i}$, for each $1 \leq i \leq m$. We assume all parameters in $\bigwedge \Gamma_{i}$ and $\bigvee \Delta_{i}$ (for $1 \leq i \leq m$) are interpreted as elements of the associated domain.

{\bf $(lift)$-rule:} Follows from Lem.~\ref{lm:persistency}.

{\bf $(\wedge_{l}), (\wedge_{r}), (\lor_{r}), (\lor_{l})$-rules:} It is not difficult to show that for each of these rules the premise, or at least one of the premises, is invalid if the conclusion is assumed invalid.


{\bf $(\imp_{r1})$-rule:} It follows from our assumption that $M, w_{m} \Vdash \bigwedge \Gamma_{m}$ and $M, w_{m} \not\Vdash \bigvee \Delta_{m} \vee A \imp B$. The latter statement implies that $M, w_{m} \not\Vdash A \imp B$, from which it follows that there exists a world $w_{m+1} \in W$ such that $Rw_{m}w_{m+1}$ and $M, w_{m+1} \Vdash A$, but $M, w_{m+1} \not\Vdash B$, letting us falsify the premise.

{\bf $(\imp_{r2})$-rule:} Our assumption implies that $M, w_{n} \Vdash \bigwedge \Gamma_{n}$, $M, w_{n} \not\Vdash \bigvee \Delta_{n} \vee A \imp B$, $M, w_{n+1} \Vdash \bigwedge \Gamma_{n+1}$, and $M, w_{n+1} \not\Vdash \bigvee \Delta_{n+1}$. The fact that $M, w_{n} \not\Vdash \bigvee \Delta_{n} \vee A \imp B$ holds implies that there exists a world $w \in W$ such that $ R w_{n} w$ and $M, w \Vdash A$ and $M, w \not\Vdash B$. Since our frames are connected, we have two cases to consider: (i) $R w w_{n+1}$, or (ii) $R w_{n+1} w$. In case (i), the left premise is falsified, and in case (ii) the right premise is falsified.

{\bf $(\imp_{l})$-rule:} Our assumption implies that $M, w_{n} \Vdash \bigwedge \Gamma_{n} \wedge A \imp B$ and $M, w_{n} \not\Vdash \bigvee \Delta_{n}$. Since $R$ is reflexive, we know that $Rw_{n}w_{n}$; this fact, in conjunction with the fact that $M, w_{n} \Vdash A \imp B$, entails that $M, w_{n} \not\Vdash A$ or $M, w_{n} \Vdash B$, which confirms that one of the premises of the rule is falsified in $M$.

{\bf $\existsl$-rule:} Our assumption implies that $M, w_{n} \Vdash \bigwedge \Gamma_{n} \wedge \exists x A$ and $M, w_{n} \not\Vdash \bigvee \Delta_{n}$. Therefore, $M, w_{n} \Vdash A[b/x]$ for some $b \in D_{w_{n}}$. Since $a$ is an eigenvariable, this implies that the premise is falsified when we interpret $a$ as $b$ at the world $w_{n}$.

{\bf $\existsr$-rule:} Similar to $\alll$ case.
\qed
\end{proof}

\begin{customthm}{\ref{thm:completeness-calc}}[Completeness of $\calc$] If $\vdash_{\ifl} A$, then $A$ is provable in $\calc$. 

\end{customthm}

\begin{proof} The claim is proven by showing that $\calc$ can derive each axiom of $\ifl$ and simulate each inference rule. We derive the quantifier axioms and all inference rules of $\ifl$, referring the reader to~\cite{Lel15} for the propositional cases.

\begin{center}
\resizebox{\columnwidth}{!}{
\begin{tabular}{c c c}
\AxiomC{$\vdash \slash \forall x A(x), A[\fvx / x] \vdash A[\fvx / x]$}
\UnaryInfC{$ \vdash  \slash \forall x A(x) \vdash A[\fvx / x]$}
\UnaryInfC{$ \vdash \forall x A(x) \imp A[\fvx / x]$}
\DisplayProof

&

\AxiomC{$ \vdash \slash A[\fvx / x] \vdash A[\fvx / x], \exists x A(x)$}
\UnaryInfC{$ \vdash \slash A[\fvx / x] \vdash \exists x A(x)$}
\UnaryInfC{$ \vdash A[\fvx / x] \imp \exists x A(x)$}
\DisplayProof

&

\AxiomC{$\vdash A[\fvx / x]$}
\RightLabel{Lem.~\ref{lm:admiss-ew}}
\dashedLine
\UnaryInfC{$\vdash \slash \vdash A[\fvx / x]$}
\UnaryInfC{$\vdash \forall x A$}
\DisplayProof
\end{tabular}
}
\end{center}
\begin{center}
\resizebox{\columnwidth}{!}{
\begin{tabular}{c}
\AxiomC{$\vdash \slash A[\fvx / x], \forall x (A(x) \vee B) \vdash B \slash A[\fvx / x] \vdash A[\fvx / x]$}
\UnaryInfC{$\vdash \slash A[\fvx / x], \forall x (A(x) \vee B) \vdash B \slash \vdash A[\fvx / x]$}

\AxiomC{$\vdash \slash B, \forall x (A(x) \vee B) \vdash B \slash \vdash A[\fvx / x]$}

\BinaryInfC{$\vdash \slash A[\fvx / x] \vee B, \forall x (A(x) \vee B) \vdash B \slash \vdash A[\fvx / x]$}
\UnaryInfC{$\vdash \slash \forall x (A(x) \vee B) \vdash B \slash \vdash A[\fvx / x]$}
\UnaryInfC{$\vdash \slash \forall x (A(x) \vee B) \vdash \forall x A(x), B$}
\UnaryInfC{$\vdash \slash \forall x (A(x) \vee B) \vdash \forall x A(x) \vee B$}
\UnaryInfC{$\vdash \forall x (A(x) \vee B) \imp \forall x A(x) \vee B$}
\DisplayProof
\end{tabular}
}
\end{center}
\begin{center}
\resizebox{\columnwidth}{!}{
\begin{tabular}{c c}
\AxiomC{$ \vdash \slash A \vdash B \slash B, A \vdash A$}
\UnaryInfC{$ \vdash \slash A \vdash B \slash B \vdash A$}

\AxiomC{$ \vdash \slash B \vdash A \slash  A, B \vdash B$}
\UnaryInfC{$ \vdash \slash B \vdash A \slash  A \vdash B$}
\UnaryInfC{$ \vdash \slash B \vdash A, A \imp B$}
\BinaryInfC{$ \vdash A \imp B \slash B \vdash A$}
\UnaryInfC{$ \vdash (A \imp B), (B \imp A)$}
\UnaryInfC{$ \vdash (A \imp B) \vee (B \imp A)$}
\DisplayProof

&

\AxiomC{$\vdash A$}
\RightLabel{Lem.~\ref{lm:admiss-ew}}
\dashedLine
\UnaryInfC{$\vdash \slash \vdash A$}

\AxiomC{$\vdash A \imp B$}
\RightLabel{Lem.~\ref{lm:invert-implies-right-one}}
\dashedLine
\UnaryInfC{$\vdash \slash A \vdash B$}

\RightLabel{Thm.~\ref{thm:cut-elimination}}
\dashedLine
\BinaryInfC{$\vdash \slash \vdash B$}
\RightLabel{Lem.~\ref{lm:admiss-mrg}}
\dashedLine
\UnaryInfC{$\vdash B$}

\DisplayProof
\end{tabular}
}
\end{center}
\begin{center}
\begin{tabular}{c} 
\AxiomC{$\Pi_{1}$}

\AxiomC{$\Pi_{2}$}

\BinaryInfC{$\vdash \slash \forall x (A(x) \imp B) \vdash \slash A[\fvx / x], \exists x A(x), \forall x (A(x) \imp B), A[\fvx / x] \imp B \vdash B$}
\UnaryInfC{$\vdash \slash \forall x (A(x) \imp B) \vdash \slash A[\fvx / x], \exists x A(x), \forall x (A(x) \imp B) \vdash B$}
\UnaryInfC{$\vdash \slash \forall x (A(x) \imp B) \vdash \slash A[\fvx / x], \exists x A(x) \vdash B$}
\UnaryInfC{$\vdash \slash \forall x (A(x) \imp B) \vdash \slash \exists x A(x) \vdash B$}
\UnaryInfC{$\vdash \slash \forall x (A(x) \imp B) \vdash \exists x A(x) \imp B$}
\UnaryInfC{$\vdash \forall x (A(x) \imp B) \imp (\exists x A(x) \imp B)$}
\DisplayProof
\end{tabular}
\end{center}

\begin{center}
\begin{tabular}{c c c}
$\Pi_{1}$

&

$= \big \{$

&

$\vdash \slash \forall x (A(x) \imp B) \vdash \slash A[\fvx / x], \exists x A(x), \forall x (A(x) \imp B), B \vdash B$
\end{tabular}

\resizebox{\columnwidth}{!}{
\begin{tabular}{c c c}
$\Pi_{2}$

&

$= \big \{$

&

$\vdash \slash \forall x (A(x) \imp B) \vdash \slash A[\fvx / x], \exists x A(x), \forall x (A(x) \imp B), A[\fvx / x] \imp B \vdash B, A[\fvx / x]$
\end{tabular}
}
\end{center}

\begin{center}
\resizebox{\columnwidth}{!}{
\begin{tabular}{c} 
\AxiomC{$\Pi_{1}'$}

\AxiomC{$\Pi_{2}'$}
\UnaryInfC{$\vdash \slash \forall x (B \imp A(x)) \vdash  \slash B, \forall x (B \imp A(x)) \vdash \slash \forall x (B \imp A(x)), B \imp A[\fvx / x] \vdash B, A[\fvx / x]$}

\BinaryInfC{$\vdash \slash \forall x (B \imp A(x)) \vdash  \slash B, \forall x (B \imp A(x)) \vdash \slash \forall x (B \imp A(x)), B \imp A[\fvx / x] \vdash A[\fvx / x]$}

\UnaryInfC{$\vdash \slash \forall x (B \imp A(x)) \vdash  \slash B, \forall x (B \imp A(x)) \vdash \slash \forall x (B \imp A(x)) \vdash A[\fvx / x]$}
\UnaryInfC{$\vdash \slash \forall x (B \imp A(x)) \vdash  \slash B \vdash \slash \vdash A[\fvx / x]$}
\UnaryInfC{$\vdash \slash \forall x (B \imp A(x)) \vdash  \slash B \vdash \forall x A(x)$}
\UnaryInfC{$\vdash \slash \forall x (B \imp A(x)) \vdash (B \imp \forall x A(x))$}
\UnaryInfC{$\vdash \forall x (B \imp A(x)) \imp (B \imp \forall x A(x))$}
\DisplayProof
\end{tabular}
}
\end{center}
\begin{center}
\resizebox{\columnwidth}{!}{
\begin{tabular}{c c c}
$\Pi_{1}'$

&

$= \big \{$

&

$\vdash \slash \forall x (B \imp A(x)) \vdash  \slash B, \forall x (B \imp A(x)) \vdash \slash \forall x (B \imp A(x)), A[\fvx / x] \vdash A[\fvx / x]$
\end{tabular}
}

\resizebox{\columnwidth}{!}{
\begin{tabular}{c c c}
$\Pi_{2}'$

&

$= \big \{$

&

$\vdash \slash \forall x (B \imp A(x)) \vdash  \slash B, \forall x (B \imp A(x)) \vdash \slash \forall x (B \imp A(x)), B \imp A[\fvx / x], B \vdash B, A[\fvx / x]$
\end{tabular}
}
\end{center}
\qed
\end{proof}

\begin{customlem}{\ref{lm:A-implies-A}} For any $A$, $\Gamma$, $\Delta$, $\g$, and $\h$, $\vdash_{\calc} \seq{\g}{\Gamma, A}{A, \Delta}{\h}$.
\end{customlem}

\begin{proof} We prove the result by induction on the complexity of $A$.

\textit{Base case.} When $A$ is atomic or $\bot$, the desired result is an instance of $\idone$ and $\botl$, respectively.

\textit{Inductive step.} We consider each case below and use IH to denote the proof given by the inductive hypothesis.

The cases where $A$ is of the form $B \land C$, $B \lor C$, or $\exists x B$ are simple, and are shown below.

\begin{center}
\begin{tabular}{c}

\AxiomC{}
\RightLabel{IH}
\dashedLine
\UnaryInfC{$\seq{\g}{\Gamma, B, C }{B, \Delta}{\h}$}
\AxiomC{}
\RightLabel{IH}
\dashedLine
\UnaryInfC{$\seq{\g}{\Gamma, B, C }{C, \Delta}{\h}$}
\BinaryInfC{$\seq{\g}{\Gamma, B, C }{B \land C, \Delta}{\h}$}
\UnaryInfC{$\seq{\g}{\Gamma, B \land C }{B \land C, \Delta}{\h}$}
\DisplayProof

\end{tabular}
\end{center}

\begin{center}
\begin{tabular}{c}

\AxiomC{}
\RightLabel{IH}
\dashedLine
\UnaryInfC{$\seq{\g}{\Gamma, B }{B,C, \Delta}{\h}$}
\AxiomC{}
\RightLabel{IH}
\dashedLine
\UnaryInfC{$\seq{\g}{\Gamma, C }{B,C, \Delta}{\h}$}
\BinaryInfC{$\seq{\g}{\Gamma, B \lor C }{B, C, \Delta}{\h}$}
\UnaryInfC{$\seq{\g}{\Gamma, B \lor C }{B \lor C, \Delta}{\h}$}
\DisplayProof

\end{tabular}
\end{center}

\begin{center}
\begin{tabular}{c}

\AxiomC{}
\RightLabel{IH}
\dashedLine
\UnaryInfC{$\seq{\g}{\Gamma, B[a/x] }{B[a/x], \exists x B, \Delta}{\h}$}
\UnaryInfC{$\seq{\g}{\Gamma, B[a/x] }{\exists x B, \Delta}{\h}$}
\UnaryInfC{$\seq{\g}{\Gamma, \exists x B}{\exists x B, \Delta}{\h}$}
\DisplayProof

\end{tabular}
\end{center}
The cases where $A$ is of the form $B \imp C$ or $\forall x B$ are a bit more cumbersome, and are explained below. We first define the linear nested sequents $\g_{i}$ (for $0 \leq i \leq n$) and $\h_{j}$
, where $\g_{0} = \g$.
$$
\g_{i} = \g \slash \Gamma_{1}, B \imp C \vdash \Delta_{1} \slash \cdots \slash \Gamma_{i}, B \imp C \vdash \Delta_{i} \quad \h_{j} = \Gamma_{j} \vdash \Delta_{j} \slash \cdots \slash \Gamma_{n} \vdash \Delta_{n}
$$
\begin{center}
\resizebox{\columnwidth}{!}{
\begin{tabular}{c} 

\AxiomC{}
\RightLabel{IH}
\dashedLine
\UnaryInfC{$\seq{\g}{\Gamma_{1}, B \imp C }{\Delta_{1} \slash B,C \vdash C}{\h_{2}}$}
\AxiomC{}
\RightLabel{IH}
\dashedLine
\UnaryInfC{$\seq{\g}{\Gamma_{1}, B \imp C }{\Delta_{1} \slash B, B \imp C \vdash B, C}{\h_{2}}$}
\BinaryInfC{$\seq{\g}{\Gamma_{1}, B \imp C }{\Delta_{1} \slash B, B \imp C \vdash C}{\h_{2}}$}
\UnaryInfC{$\seq{\g}{\Gamma_{1}, B \imp C }{\Delta_{1} \slash B \vdash C}{\h_{2}}$}

\AxiomC{$\Pi_{0}$}

\BinaryInfC{$\seq{\g}{\Gamma_{1}, B \imp C }{B \imp C, \Delta_{1}}{\h_{2}}$}
\DisplayProof

\end{tabular}
}
\end{center}
The derivations $\Pi_{i}$, with $0 \leq i \leq n-3$, are as follows:
\begin{center}
\resizebox{\columnwidth}{!}{
\begin{tabular}{c} 
\AxiomC{}
\RightLabel{IH}
\dashedLine
\UnaryInfC{$\seq{\g_{i+1}}{\Gamma_{i+2}, B \imp C }{\Delta_{i+2} \slash B,C \vdash C}{\h_{i+3}}$}
\AxiomC{}
\RightLabel{IH}
\dashedLine
\UnaryInfC{$\seq{\g_{i+1}}{\Gamma_{i+2}, B \imp C }{\Delta_{i+2} \slash B, B \imp C \vdash B, C}{\h_{i+3}}$}
\BinaryInfC{$\seq{\g_{i+1}}{\Gamma_{i+2}, B \imp C }{\Delta_{i+2} \slash B, B \imp C \vdash C}{\h_{i+3}}$}
\UnaryInfC{$\seq{\g_{i+1}}{\Gamma_{i+2}, B \imp C }{\Delta_{i+2} \slash B \vdash C}{\h_{i+3}}$}

\AxiomC{$\Pi_{i+1}$}
\BinaryInfC{$\seq{\g_{i+1}}{\Gamma_{i+2}, B \imp C }{B \imp C, \Delta_{i+2}}{\h_{i+3}}$}
\RightLabel{=}
\dottedLine
\UnaryInfC{$\seq{\g_{i}}{\Gamma_{i+1}, B \imp C }{\Delta_{i+1} \slash \Gamma_{i+2}, B \imp C \vdash B \imp C, \Delta_{i+2}}{\h_{i+3}}$}
\UnaryInfC{$\seq{\g_{i}}{\Gamma_{i+1}, B \imp C }{\Delta_{i+1} \slash \Gamma_{i+2} \vdash B \imp C, \Delta_{i+2}}{\h_{i+3}}$}
\DisplayProof
\end{tabular}
}
\end{center}
Last, the derivation $\Pi_{n-2}$ is as follows:
\begin{center}
\begin{tabular}{c} 
\AxiomC{}
\RightLabel{IH}
\dashedLine
\UnaryInfC{$\g_{n} \slash B, C \vdash C$}

\AxiomC{}
\RightLabel{IH}
\dashedLine
\UnaryInfC{$\g_{n} \slash B, B \imp C \vdash B, C$}

\BinaryInfC{$\g_{n} \slash B, B \imp C \vdash C$}
\RightLabel{=}
\dottedLine

\UnaryInfC{$\g_{n-2} \slash \Gamma_{n-1}, B \imp C \vdash \Delta_{n-1} \slash \Gamma_{n}, B \imp C \vdash \Delta_{n} \slash B, B \imp C \vdash C$}

\UnaryInfC{$\g_{n-2} \slash \Gamma_{n-1}, B \imp C \vdash \Delta_{n-1} \slash \Gamma_{n}, B \imp C \vdash \Delta_{n} \slash B \vdash C$}
\UnaryInfC{$\g_{n-2} \slash \Gamma_{n-1}, B \imp C \vdash \Delta_{n-1} \slash \Gamma_{n} \vdash \Delta_{n} \slash B \vdash C$}
\UnaryInfC{$\g_{n-2} \slash \Gamma_{n-1}, B \imp C \vdash \Delta_{n-1} \slash \Gamma_{n} \vdash B \imp C, \Delta_{n}$}
\DisplayProof
\end{tabular}
\end{center}

Let us consider the case where $A$ is of the form $\forall x B$. We first define the linear nested sequents $\g_{i}$ (for $0 \leq i \leq n$) and $\h_{j}$
, where $\g_{0} = \g$.
$$
\g_{i} = \g \slash \Gamma_{1}, \forall x B \vdash \Delta_{1} \slash \cdots \slash \Gamma_{i}, \forall x B \vdash \Delta_{i} \quad \h_{j} = \Gamma_{j} \vdash \Delta_{j} \slash \cdots \slash \Gamma_{n} \vdash \Delta_{n}
$$

\begin{center}
\begin{tabular}{c} 

\AxiomC{}
\RightLabel{IH}
\dashedLine
\UnaryInfC{$\seq{\g}{\Gamma_{1}, \forall x B }{\Delta_{1} \slash A[y/x],\forall x A \vdash A[y/x]}{\h_{2}}$}
\UnaryInfC{$\seq{\g}{\Gamma_{1}, \forall x B }{\Delta_{1} \slash \forall x A \vdash A[y/x]}{\h_{2}}$}
\UnaryInfC{$\seq{\g}{\Gamma_{1}, \forall x B }{\Delta_{1} \slash \vdash A[y/x]}{\h_{2}}$}

\AxiomC{$\Pi_{0}$}

\BinaryInfC{$\seq{\g}{\Gamma_{1}, \forall x B }{\forall x B, \Delta_{1}}{\h_{2}}$}
\DisplayProof

\end{tabular}
\end{center}
The derivations $\Pi_{i}$, with $0 \leq i \leq n-3$, are as follows:
\begin{center}
\resizebox{\columnwidth}{!}{
\begin{tabular}{c}
\AxiomC{}
\RightLabel{IH}
\dashedLine
\UnaryInfC{$\seq{\g_{i+2}}{\forall x B, B[y_{i}/x] }{B[y_{i}/x]}{\h_{i+3}}$}
\UnaryInfC{$\seq{\g_{i+2}}{\forall x B }{B[y_{i}/x]}{\h_{i+3}}$}
\RightLabel{=}
\dottedLine
\UnaryInfC{$\seq{\g_{i}}{\Gamma_{i+1}, \forall x B }{\Delta_{i+1}}{\Gamma_{i+2}, \forall x B \vdash \Delta_{i+2} \slash \forall x B \vdash B[y_{i}/x] \slash \h_{i+3}}$}
\UnaryInfC{$\seq{\g_{i}}{\Gamma_{i+1}, \forall x B }{\Delta_{i+1}}{\Gamma_{i+2}, \forall x B \vdash \Delta_{i+2} \slash \vdash B[y_{i}/x] \slash \h_{i+3}}$}

\AxiomC{$\Pi_{i+1}$}

\BinaryInfC{$\seq{\g_{i}}{\Gamma_{i+1}, \forall x B }{\Delta_{i+1}}{\Gamma_{i+2}, \forall x B \vdash \forall x B, \Delta_{i+2} \slash \h_{i+3}}$}
\UnaryInfC{$\seq{\g_{i}}{\Gamma_{i+1}, \forall x B }{\Delta_{i+1}}{\Gamma_{i+2} \vdash \forall x B, \Delta_{i+2} \slash \h_{i+3}}$}
\DisplayProof
\end{tabular}
}
\end{center}
The last component of the derivation, $\Pi_{n-2}$ is given below:
\begin{center}
\begin{tabular}{c}
\AxiomC{}
\RightLabel{IH}
\dashedLine
\UnaryInfC{$\seq{\g_{n-2}}{\Gamma_{n-1}, \forall x B }{\Delta_{n-1}}{\Gamma_{n}, \forall x B \vdash \Delta_{n} \slash \forall x B, B[y_{n} / x] \vdash B[y_{n} / x]}$}
\UnaryInfC{$\seq{\g_{n-2}}{\Gamma_{n-1}, \forall x B }{\Delta_{n-1}}{\Gamma_{n}, \forall x B \vdash \Delta_{n} \slash \forall x B \vdash B[y_{n} / x]}$}
\UnaryInfC{$\seq{\g_{n-2}}{\Gamma_{n-1}, \forall x B }{\Delta_{n-1}}{\Gamma_{n}, \forall x B \vdash \Delta_{n} \slash \vdash B[y_{n} / x]}$}
\UnaryInfC{$\seq{\g_{n-2}}{\Gamma_{n-1}, \forall x B }{\Delta_{n-1}}{\Gamma_{n} \vdash \Delta_{n} \slash \vdash B[y_{n} / x]}$}
\UnaryInfC{$\seq{\g_{n-2}}{\Gamma_{n-1}, \forall x B }{\Delta_{n-1}}{\Gamma_{n} \vdash \forall x B, \Delta_{n}}$}
\DisplayProof
\end{tabular}
\end{center}
\qed
\end{proof}

\begin{customlem}{\ref{lm:substitution-lemma}}
The $\vsub$ rule is hp-admissible in $\calc$.
\end{customlem}

\begin{proof} We prove the result by induction on the height of the given derivation of $\g$.

\textit{Base case.} Any instance of the rule $\idone$, $\idtwo$, or $\botl$ is still an instance of the rule under the variable substitution $[\fvy /\fvx]$.

\textit{Inductive step.} For 
all rules, with the exception of the $\allrone$, $\existsl$, and $\allrtwo$ rules, the claim follows straightforwardly by applying IH followed by the corresponding rule. The nontrivial cases occur when the last rule applied is an instance of $\allrone$, $\existsl$, or $\allrtwo$, and the variable substituted into the conclusion of the inference is also the eigenvariable of the inference:

\begin{center}
\begin{tabular}{c @{\hskip 1em} c}
\AxiomC{$\seq{\g}{\Gamma}{\Delta}{ \vdash A[\fvx/x] }$}
\UnaryInfC{$\g \sslash \Gamma \vdash \Delta, \forall x A$}
\UnaryInfC{$(\g \sslash \Gamma \vdash \Delta, \forall x A)[\fvx / \fvy]$}
\DisplayProof

&

\AxiomC{$\seq{\g}{\Gamma, A[\fvx/x]}{\Delta}{\h}$}
\UnaryInfC{$\seq{\g}{\Gamma, \exists x A}{\Delta}{\h}$}
\UnaryInfC{$(\seq{\g}{\Gamma, \exists x A}{\Delta}{\h})[\fvx / \fvy]$}
\DisplayProof
\end{tabular}
\end{center}

\begin{center}
\AxiomC{${\g \sslash \Gamma_{1} \vdash \Delta_{1} \sslash \vdash A[\fvx/x] \sslash \Gamma_{2} \vdash \Delta_{2} \sslash \h}$}
\AxiomC{${\g \sslash \Gamma_{1} \vdash \Delta_{1} \sslash \Gamma_{2} \vdash \Delta_{2}, \forall x A \sslash \h}$}
\BinaryInfC{$\g \sslash \Gamma_{1} \vdash \Delta_{1}, \forall x A \sslash \Gamma_{2} \vdash \Delta_{2} \sslash \h$}
\UnaryInfC{$(\g \sslash \Gamma_{1} \vdash \Delta_{1}, \forall x A \sslash \Gamma_{2} \vdash \Delta_{2} \sslash \h)[\fvx / \fvy]$}
\DisplayProof
\end{center}

In such cases we invoke the inductive hypothesis twice: first, we apply the substitution $[\fvz / \fvx]$ to the premise, where $\fvz$ is a fresh parameter, and then we invoke the inductive hypothesis again and apply the substitution $[\fvx / \fvy]$. The desired result follows by a single application of each rule.
\qed
\end{proof}


\begin{customlem}{\ref{lm:hp-admiss-iw}}
The $\iw$ rule is hp-admissible in $\calc$.
\end{customlem}

\begin{proof} We know by~\cite[Lem.~5.5]{LelKuz18} that $\iw$ is admissible in $\lng$. We extend the argument to $\calc$ and consider the cases of permuting $\iw$ past the $\allrone$, $\existsl$, and $\allrtwo$ rules; the $\lift$, $\alll$, $\existsr$ cases are trivial.

Suppose we have a $\allrone$ inference followed by an instance of $\iw$:

\begin{center}
\begin{tabular}{c}

\AxiomC{$\g \slash \Gamma \vdash \Delta \slash \vdash A[\fvx/x]$}
\RightLabel{$\allrone$}
\UnaryInfC{$\g \slash \Gamma \vdash \Delta, \forall x A$}
\RightLabel{$\iw$}
\UnaryInfC{$\g' \slash \Gamma' \vdash \Delta', \forall x A$}
\DisplayProof

\end{tabular}
\end{center}

The nontrivial case occurs when the $\iw$ rule weakens in a formula containing the parameter $\fvx$. If this happens to be the case, then we invoke Lem.~\ref{lm:substitution-lemma} and apply a substitution $[\fvy/\fvx]$ where $\fvy$ is a fresh parameter not occurring in the derivation above. After performing this operation, we may complete the derivation as shown below:
\begin{center}
\begin{tabular}{c}

\AxiomC{$\g \slash \Gamma \vdash \Delta \slash \vdash A[\fvx/x]$}
\dashedLine
\RightLabel{Lem.~\ref{lm:substitution-lemma}}
\UnaryInfC{$\g[\fvy/\fvx] \slash \Gamma[\fvy/\fvx] \vdash \Delta[\fvy/\fvx] \slash \vdash A[\fvx/x][\fvy/\fvx]$}
\RightLabel{=}
\dottedLine
\UnaryInfC{$\g\slash \Gamma\vdash \Delta \slash \vdash A[\fvy/x]$}
\RightLabel{$\iw$}
\UnaryInfC{$\g' \slash \Gamma' \vdash \Delta' \slash \vdash A[\fvy/x]$}
\RightLabel{$\allrone$}
\UnaryInfC{$\g' \slash \Gamma' \vdash \Delta', \forall x A$}
\DisplayProof

\end{tabular}
\end{center}
The second and third line are equal because $\fvx$ does not occur in $\g$, $\Gamma$, or $\Delta$ (i.e., $\fvx$ is an eigenvariable) and also because $A[\fvx/x][\fvy/\fvx] = A[\fvy/x]$ (we may assume w.l.o.g. that $A$ does not contain any occurrences of $a$). Last, we may apply the $\allrone$ rule since we are guaranteed that $\fvy$ is an eigenvariable by choice.

The cases for $\existsl$ and $\allrtwo$ are shown similarly.
\qed
\end{proof}

\setcounter{equation}{0}

\begin{customlem}{\ref{lm:m-invertibility}} If $\sum_{i=1}^{n} k_{n} \geq 1$, then
\begin{center}
\begin{tabular}{c @{\hskip 1em} c @{\hskip 1em} c}
$
\begin{aligned}
& (i) \ (1) \text{ implies } (2)\\
& (ii) \ (3) \text{ implies } (4) \text{ and } (5)
\end{aligned}
$

&

$
\begin{aligned}
& (iii) \ (6) \text{ implies } (7) \text{ and } (8)\\
& (iv) \ (9) \text{ implies } (10)
\end{aligned}
$

&

$
\begin{aligned}
& (v) \ (11) \text{ implies } (12)\\
\
\end{aligned}
$
\end{tabular}
\end{center}
\begin{eqnarray}
\vdash_{\calc} \ \Gamma_{1}, (A \land B)^{k_{1}} \vdash \Delta_{1} \slash & \cdots & \slash \Gamma_{n}, (A \land B)^{k_{n}} \vdash \Delta_{n}\\
\vdash_{\calc} \ \Gamma_{1}, A^{k_{1}}, B^{k_{1}} \vdash \Delta_{1} \slash & \cdots & \slash \Gamma_{n}, A^{k_{n}}, B^{k_{n}} \vdash \Delta_{n}\\
\vdash_{\calc} \ \Gamma_{1}, (A \lor B)^{k_{1}} \vdash \Delta_{1} \slash & \cdots & \slash \Gamma_{n}, (A \lor B)^{k_{n}} \vdash \Delta_{n}\\
\vdash_{\calc} \ \Gamma_{1}, A^{k_{1}} \vdash \Delta_{1} \slash & \cdots & \slash \Gamma_{n}, A^{k_{n}} \vdash \Delta_{n}\\
\vdash_{\calc} \ \Gamma_{1}, B^{k_{1}} \vdash \Delta_{1} \slash & \cdots & \slash \Gamma_{n}, B^{k_{n}} \vdash \Delta_{n}\\
\vdash_{\calc} \ \Gamma_{1}, (A \imp B)^{k_{1}} \vdash \Delta_{1} \slash & \cdots & \slash \Gamma_{n}, (A \imp B)^{k_{n}} \vdash \Delta_{n}\\
\vdash_{\calc} \ \Gamma_{1}, B^{k_{1}} \vdash \Delta_{1} \slash & \cdots & \slash \Gamma_{n}, B^{k_{n}} \vdash \Delta_{n}\\
\vdash_{\calc} \ \Gamma_{1}, (A \imp B)^{k_{1}} \vdash \Delta_{1}, A^{k_{1}} \slash & \cdots & \slash \Gamma_{n}, (A \imp B)^{k_{n}} \vdash \Delta_{n}, A^{k_{n}}\\
\vdash_{\calc} \ \Gamma_{1}, (\forall x A)^{k_{1}} \vdash \Delta_{1} \slash & \cdots & \slash \Gamma_{n}, (\forall x A)^{k_{n}} \vdash \Delta_{n}\\
\vdash_{\calc} \ \Gamma_{1}, A[\fvx/x]^{k_{1}}, (\forall x A)^{k_{1}} \vdash \Delta_{1} \slash & \cdots & \slash \Gamma_{n}, A[\fvx/x]^{k_{n}}, (\forall x A)^{k_{n}} \vdash \Delta_{n}\\
\vdash_{\calc} \ \Gamma_{1}, (\exists x A)^{k_{1}} \vdash \Delta_{1} \slash & \cdots & \slash \Gamma_{n}, (\exists x A)^{k_{n}} \vdash \Delta_{n}\\
\vdash_{\calc} \ \Gamma_{1}, A[\fvx/x]^{k_{1}} \vdash \Delta_{1} \slash & \cdots & \slash \Gamma_{n}, A[\fvx/x]^{k_{n}}\vdash \Delta_{n}
\end{eqnarray}

\end{customlem}

\begin{proof} By~\cite[Lem.~5.9]{LelKuz18} we know that claims (i)-(iii) hold for $\lng$. It is easy to show that the proof of~\cite[Lem.~5.9]{LelKuz18} can be extended to the quantifier rules for claims (i) and (ii) by applying the IH and then the rule. We therefore argue that claim (iii) continues to hold in the presence of the quantifier rules, and also argue that claims (iv) and (v) hold.

\textit{Claim (iii).} We know that (6) implies (8) by Lem.~\ref{lm:hp-admiss-iw}. One can prove that (6) implies (7) by induction on the height of the given derivation. By~\cite[Lem.~5.9]{LelKuz18} all cases with the exception of the quantifier rules hold. 
All of the quantifier cases are handled by applying IH followed by an application of the corresponding rule.

\textit{Claim (iv).} Statement (9) implies (10) by Lem.~\ref{lm:hp-admiss-iw}.

\textit{Claim (v).} We argue that (11) implies (12) by induction on the height of the given derivation of $\g = \Gamma_{1}, (\exists x A)^{k_{1}} \vdash \Delta_{1} \slash \cdots \slash \Gamma_{n}, (\exists x A)^{k_{n}} \vdash \Delta_{n}$.

\textit{Base case.} If $\g$ is the result of $\idone$, $\idtwo$, or $\botl$, then $\Gamma_{1}, A[\fvx / x]^{k_{1}} \vdash \Delta_{1} \slash \cdots \slash \Gamma_{n}, A[\fvx / x]^{k_{n}} \vdash \Delta_{n}$ is an instance of the corresponding rule as well.

\textit{Inductive step.} For all rules, with the exception of $\allrone$, $\existsl$, and $\allrtwo$, we apply IH to the premise(s) followed by the rule. In the $\allrone$, $\existsl$, and $\allrtwo$ cases we must ensure that the eigenvariable of the inference is not identical to the parameter $\fvx$ occurring in $A[\fvx / x]$; however, due to Lem.~\ref{lm:substitution-lemma}, this can always be ensured.
\qed
\end{proof}

\label{appendix}

\end{document}